\documentclass[11pt,a4paper]{article}
\usepackage[utf8]{inputenc}
\usepackage{amsmath, amsthm}
\usepackage{amsfonts}
\usepackage{dsfont}
\usepackage{amssymb}
\usepackage{subcaption}
\usepackage{graphicx, color, xcolor, fancybox}
\usepackage{ulem, comment, multirow, natbib, array, booktabs}
\usepackage[most]{tcolorbox}
\usepackage[left=2cm,right=2cm,top=2cm,bottom=2cm]{geometry}
\usepackage{hyperref}

\newcommand{\R}{\mathbb{R}}

\newcommand{\id}{\mathds{1}}

\newcommand{\bk}{\boldsymbol{k}}
\newcommand{\bphi}{\boldsymbol{\varphi}}
\newcommand{\ind}{\mathbf{1}}

\DeclareMathOperator{\Prb}{\mathbb{P}}

\DeclareMathOperator{\argmin}{argmin}

\allowdisplaybreaks[1]

\newtheorem{theorem}{Theorem}
\newtheorem{proposition}{Proposition}

\newcommand{\usup}[1]{\,\underset{#1}{\sup}\,}
\def\tvc#1{\color{black}#1 \color{black}}
\newcommand{\tani}[1]{{\color{black}#1}}
\newcommand{\vince}[1]{{\color{black}#1}}

\title{Goodness-of-fit testing for the stationary density of a size-structured PDE}

\author{Van H\`a Hoang\thanks{Faculty of Mathematics and Computer Science, Vietnam National University - Ho Chi Minh City, Viet Nam; LAMA, Univ Gustave Eiffel, UPEM, Univ Paris Est Creteil, CNRS, F-77447, Marne-la-Vallée, France; \texttt{E-mail: hvha@hcmus.edu.vn }}, \quad 
Ph\'u Thanh Nguyen\thanks{Faculty of Mathematics and Computer Science, Vietnam National University - Ho Chi Minh City, Viet Nam; \texttt{E-mail: npthanh@hcmus.edu.vn }}, \quad 
Thanh Mai Pham Ngoc\thanks{LAGA, CNRS, UMR 7539, Universit\'e Sorbonne Paris Nord, Villetaneuse, France; E-mail: \texttt{phamngoc@math.univ-paris13.fr}},\\
Vincent Rivoirard\thanks{Ceremade, CNRS, UMR 7534, Universit\'e Paris-Dauphine, PSL Research University, 75016 Paris,
France; E-mail: \texttt{Vincent.Rivoirard@dauphine.fr}}, \quad
Viet Chi Tran\thanks{Univ. Lille, CNRS, Inria, UMR 8524 - Laboratoire Paul Painlevé, F-59000 Lille, France; LAMA, Univ Gustave Eiffel, UPEM, Univ Paris Est Creteil, CNRS, F-77447, Marne-la-Vallée, France; E-mail: \texttt{viet-chi.tran@inria.fr}}}

\date{\today}
\begin{document}
\maketitle

\begin{abstract}
We consider two division models for structured cell populations, where cells can grow, age and divide. These models have been introduced in the literature under the denomination of `mitosis' and `adder' models. In the recent years, there has been an increasing interest in biology to understand whether the cells divide equally or not, as this can be related to important mechanisms in cellular aging or recovery. We are therefore interested in testing the null hypothesis $H_0$ \tvc{where the division of a mother cell results into two daughters of equal size}, against the alternative hypothesis $H_1$ where the division is asymmetric and ruled by a kernel that is absolutely continuous with respect to the Lebesgue measure. The sample consists of i.i.d. observations of cell sizes and ages drawn from the population, and the division is not directly observed. The hypotheses of the test are reformulated as hypotheses on the stationary size and age distributions of the models, which we assume are also the distributions of the observations. We propose a goodness-of-fit test that we study numerically on simulated data before applying it on real data.
\end{abstract}


\noindent AMS2020: 62G10.

\bigskip

\section{Introduction}

The question of whether cells or bacteria divide exactly into two similar daughters or whether the division is asymmetric has been of interest to researchers in biology and mathematics for several years, \tvc{starting from} \cite{Stewart05,Evans07,lindnermaddendemarezstewarttaddei}.
For \textit{Escherichia coli}, asymmetric division, which segregates damage at the expense of one of the two daughters, has been suggested as a strategy for population maintenance and an explanation for aging in these organisms. A simple model considers a population structured by size. A cell or bacterium of size $x\in  \R_+$ grows with speed 1 and divides at a rate $B(x)$ into two cells of sizes $\theta x$ and $(1-\theta)x\geq 0$ where the fraction $\theta$ is a random variable of distribution $\kappa(d\theta)$ with values in $[0,1]$. \tvc{The division is symmetric in the sense that $\theta \stackrel{(d)}{=} 1-\theta$ in distribution.} The evolution of the size distribution in a large population can be described by the following partial differential equation (PDE), where \tvc{$n(t,x)$ represents the density of the population with trait $x$ at time $t$}:  

\begin{align}
    & \partial_t n(t,x)+\partial_x n(t,x)+B(x)n(t,x)=2 \int_0^1 B\left(\frac{x}{\theta}\right)n\left(t,\frac{x}{\theta}\right) \tvc{\theta^{-1}} \kappa(d\theta)\label{EDP1}\\
    & n(0,x)=n_0(x),\quad n(t,0)=0. \notag 
\end{align}
\vince{In the previous equation, $n(t,x)$ represents the number of cells of size $x$ at time $t$. The process by which a cell divides is represented by the difference $2 \int_0^1 B\left(\frac{x}{\theta}\right)n\left(t,\frac{x}{\theta}\right)\theta^{-1}\kappa(d\theta)-B(x)n(t,x)$, the constant 2 means that a mother cell gives birth to two daughters. Observe that if $\kappa(d\theta)=\delta_{1/2}(\theta)$ it means that the two daughters have exactly the same size. In this case, the right hand side of (1) is equal to    $4 B(2x)n(t,2x)$. The introduction of this equation dates back from \cite{BELL1967329} and \cite{ANDERSON1967353}.
  Such equations have been studied by \cite{MR764484} and then more recently by \cite{Doumic09, Perthame07}. We refer to \cite{HoangPhamNgocRivoirardTran2020} for a derivation from the individual-based dynamics described in the first paragraph.}

In this paper, we also consider the `adder' model introduced in \cite{taheriaraghi_etal,hallwakegandar}, see also \cite{Doumic2018}.
In addition to size $x$, the population is also categorized by the increase in size $a=x-y$ between the current size $x$ and the size at birth $y$. The increment $a$ can be viewed as a kind of age that increases monotonically and starts at zero at birth. The evolution equations replacing \eqref{EDP1} are then

\begin{align} 
    & \partial_t n(t,x,a)+\partial_x \big( g(x) n(t,x,a)\big) + \partial_a \big(g(x)n(t,x,a)\big)+g(x) B(a)n(t,x,a)=0,\label{EDP:adder}\\
    & g(x)n(t,x,0)=2 \int_0^1 g\big(\frac{x}{\theta}) \int_0^{\infty} B(a) n\big(t,\frac{x}{\theta},a\big) \ da \ \theta^{-1} \kappa(d\theta) \nonumber\\
    & n(0,x,a)=n_0(a,x),\quad g(0) n(t,0,a)=0. \nonumber
\end{align}
The term  $g(x)$ denotes the growth rate of a cell of size $x$. The division rate appears in this model in the form of a product $g(x)B(a)$, which can be interpreted as $B$ being the division rate in a unit of growth instead of a unit of time.  %\vinc{Peut-etre dire un mot sur $g$ qui n'apparaissait pas dans  \eqref{EDP1} comme cela c'est moins mysterieux.} 
The cell somehow ignores the ``time" and uses its growth as a clock. With the new single-cell data, some authors have observed that, for some bacteria, experiments favour the adder model over size-structured equations such as \eqref{EDP1}, see e.g. \cite{Sauls2016,taheriaraghi_etal,Camposetal2014}. It is one of the aims of this paper to base such observations on a methodology that is statistically sound. Statistical models for inferring the division kernels were carried out in \cite{Bourgeron14,Hoang2017,HoangPhamNgocRivoirardTran2020,DoumicEscobedoTournus} for example.

\tani{The goal of the present work is to determine whether the underlying mechanism governing cell division is described by a deterministic cut model in $1/2$ (hypothesis $H_0$) or not. %by a random cut model admitting a probability density (hypothesis $H_1$). To do so we have to tackle a goodness-of-fit testing of the null hypothesis $H_0$ where the two daughters are identical against the alternative hypothesis $H_1$ where the division is ruled by a kernel admitting a density with respect to the Lebesgue measure on $[0,1]$. 
The statistical test is formulated as follows:
$H_0\ :\  \kappa(d\theta)=\delta_{1/2}(d\theta)$ against  $H_1\ :  \kappa(d\theta)\neq \delta_{1/2}(d\theta)$. Thus we are led to construct a goodness-of-fit test. 
}\tvc{ In this paper, we are interested in constructing these tests from i.i.d. observations drawn from the population distribution, corresponding for example to wild populations. Mathematically, this involves an inverse problem, as sizes are observed but not the division ratio. In \cite{Hoang2017}, an estimation of the division kernel was obtained from the direct observation of the sizes of mothers and related daughters. Datasets providing such information are now available from microfluidic Mother Machine devices (non wild populations) \cite{siletreutsaulsvadialevinjun,witzvannumwegenjulou,tiruvadikrishnanmannikkarlinamirmannik,perrindoumicelkarouimeleard}. This allows direct computation of an estimator of $\kappa$ as in \cite{Hoang2017}, and leads to different (and more usual) testing tools. The population approach developed here can be can also be useful beyond \textit{E. coli} when datasets may not be available. Subjects involving divisions arise for stem cells, cancer cell or ovarian follicle cells to name a few, e.g. \cite{ballifclementyvinec,fernandezbarandabansayemounierugomeleardgiraudier}. 
}

%We will assume in this paper that under $H_1$, the alternative hypothesis, the kernel $\kappa(d\theta)$ admits a density $h(\theta)$ with respect to the Lebesgue measure on $[0,1]$, so that the alternative hypothesis writes:
%an alternative where $\kappa(d\theta)$ is a symmetric distribution on $[0,1]$ centered at $1/2$ (considering the daughters as exchangeable and labeled independently from their sizes). 

In the statistical literature, standard goodness-of-fit testing treats the case of direct observations and one-dimensional density model. More precisely,  one wants to test the null hypothesis  $``f \in \mathcal{F}"$ against the alternative $``f \notin \mathcal{F}"$ (where $\mathcal{F}$ typically equals some specified density $f_0$) from observations drawn with respect to the density $f$. Typically, a null hypothesis corresponds to our belief that the observed data are organized
in a relatively simple way, which means that the structure of the underlying model is completely specified. A test $\phi$ is a rule to accept or to reject the null hypothesis by means of the observations. It is a measurable function of the observations taking values in the two points set $\{0, 1 \}$.  The value  $0$ is treated as accepting $H_0$, and the value $1$ means that
the test rejects $H_0$. If one assumes that $f \in L^2(\mathbb{R})$, the construction of the test statistics will be based on an estimate of the distance $\int_\R (f-f_0)^2$. 
 %It is natural to measure the quality of any test by its sensitivity to perturbations or
%contaminations of this "null" model.  
The quality of any test is measured by the probabilities of the corresponding errors.  The probability of error of the first kind is the probability under the null to reject the hypothesis. Now if $f$ is a function from the alternative set the probability of error of the second kind at $f$ is defined as $\mathbb{P}_f(\phi=0)$ and the value $1- \mathbb{P}_f(\phi=0)$ is called the power of the test $\phi $ at $f$. One rejects the null hypothesis if the test statistics is larger than some threshold calibrated such that the error of first kind of the test does not exceed some desired value (i.e the level of the test). 

Standard goodness-of-fit testing has been widely studied since the famous Kolmogorov Smirnov and Cramer von Mises tests based on the empirical distribution function.  For  modern theoretical studies in a nonparametric setting (one does not assume any special parametric structure for the alternative which is expressed in terms of some smoothness classes such as Besov or H\"older classes), we refer to the pioneered work of \cite{ingster,spokoiny} or  \cite{fromontlaurent2006} who propose adaptive minimax testing procedures. Regarding the standard framework of \tani{goodness-of-fit testing},  our statistical setting induced by our biological problem is original in many ways. First, our observations are indirect and \tani{are not distributed according the distribution $\kappa(d\theta)$}. Indeed, we are dealing with a non standard inverse problem induced by the PDEs \eqref{EDP1} and \eqref{EDP:adder} (for details, see Equation (12) and related comments in \cite{HoangPhamNgocRivoirardTran2020}) and what is more with a two dimensional goodness-of fit test in the case of the adder model.  Secondly, our null hypothesis does not refer to a density but a Dirac distribution in $\frac 1 2$.
And when considering literature about goodness-of fit testing, very few studies have been devoted to the case of indirect observations. Indirect observations involve some operator to be inverted which might lead to great instability in estimating the unknown density. For goodness-of-fit tests with indirect observations, let us quote the works of \cite{Bissantz2009} for the inverse regression
problem, \cite{Holzmann2007}, \cite{Laurent2011}, \cite{Butucea2007}, \cite{Butucea2009}, \cite{LacourPhamNgoc2014} for the classical convolution model.   In the sequel, we shall see how to overcome those difficulties by considering the stationary density arising from the equations governing the biological models.

In Section \ref{sec:modeles}, we detail the two models introduced above, namely the size-structured mitosis and the adder models. The goodness-of-fit tests are described in Section \ref{sec:goodness-of-fit}. In Section \ref{simulations}, we perform the tests on simulated data first to check their properties. Then, in Section \ref{sec:Ecoli}, we deal with a real dataset and discuss the modelling of \textit{E. coli} data from \cite{Bourgeron14,Stewart05}.

\section{Populations of dividing cells}\label{sec:modeles}

 \tvc{The equations \eqref{EDP1} and \eqref{EDP:adder} can be obtained from stochastic individual-based models where the dynamics are described at the level of cells or bacteria.} Let us begin with a presentation of the two models considered in the paper, namely a mitosis model for size-structured populations and the `adder' model introduced by \cite{taheriaraghi_etal,hallwakegandar}.

\subsection{A generalized mitosis model for size-structured cell populations}\label{sec:modele-mitose}

The mitosis model associated to \eqref{EDP1} corresponds to the large population limit in populations of cells or bacteria that divide at rates that depend on their sizes. \tvc{In this population, individuals' sizes increase in time with speed 1.} An individual with size $x$ divides into two daughter cells at rate $B(x)$ such that 
\begin{equation}
    \int^{+\infty} B(x)\ dx=+\infty,
\end{equation}meaning that the remaining lifetime of this individual has density
\[f(a)=B(x+a) \exp\Big(\int_0^a B(x+\alpha) \ d\alpha \Big).\]
Upon division, the mother of size $x$ is replaced by two daughters of random size $\theta x$ and $(1-\theta)x$ where $\theta$ is drawn in the probability distribution $\kappa(d\theta)$. \tvc{This defines a sequence of random events, allowing the description of the dynamics at the level of the individuals. The associated stochastic process is detailed in \cite{FerriereTran09,BansayeMeleard,Hoang2017} and is called \textit{individual based model}. If the times of these events and the associated divisions (namely the fractions $\theta$) were all observed, it would be possible to infer and test directly whether $\kappa=\delta_{1/2}$. An estimator based on such data has been studied in \cite{Hoang2017}. As for the test, if the fractions $\theta$ are observed, $H_0$ would be rejected as soon as a division producing daughters of different sizes is observed. Here, we assume that we do not have such data, but only the observation of sizes of individuals drawn from the population independently. This is why we add the extra assumptions that our population is large and close to an `equilibrium' state. Then, the size distribution can be approximated by mean of the eigenelements of the operator associated with \eqref{EDP1}. 
In the case where the division rate is constant, $B(x)=R$, considered by \cite{HoangPhamNgocRivoirardTran2020}, it is shown that the correctly rescaled solution of the limiting PDE converges in long time to a stationary distribution that we denote by $N(x)$ and that solves:
\begin{equation}\label{eq:model}
 D(x) + 2R\  N(x) = 2R \int_0^1 N\left(\frac{x}{\theta}\right) \frac{ \kappa (d\theta)}{\theta}, \quad x \ge 0,
\end{equation}
where $D(x):= \partial_x N(x)$. See \cite[Proposition 3]{HoangPhamNgocRivoirardTran2020}.} Our goal is to test the null hypothesis $H_0: \kappa(d\theta) = \delta_{1/2}(d\theta)$ against the alternative $H_1: \kappa(d\theta) \neq \delta_{1/2}(d\theta)$, to see whether we can assume that the division creates two equal daughters.

In the sequel, we will not assume that the fractions $\theta$ are observed. The data rather consists in an $n$-sample $(X_1,\dots, X_n)$ drawn from the stationary density $N$.
Therefore, instead of testing directly $H_0: \kappa = \delta_{1/2}$ by estimating the division kernel $\kappa(d\theta)$ that is not always absolutely continuous with respect to the Lebesgue measure, which would raise an intricate inverse problem, we establish that we can perform the test on the stationary function $N$ which is more natural under this sampling scheme.

\subsection{An adder model for populations with size homeostasis divisions} \label{sec:adder-model}

An alternative model to the mitosis size-structured equations presented in Section \ref{sec:modele-mitose} is the `adder' model presented in \eqref{EDP:adder}, also named as incremental model and introduced by \cite{taheriaraghi_etal,hallwakegandar}. In some cases, see e.g. \cite{Sauls2016,Camposetal2014}, \tvc{data point to a model with size homeostasis: the division does not occur at a critical size (`sizer') or after a specific time laps (`timer') but rather is triggered by the size added since the last division.}

In the `adder' model, cells are described by their size-increment and their size, respectively denoted by $a$ and $x$ in the sequel. The growth rate $g(x)$ only depends on size, and the division rate is of the form $B(a)g(x)$. The latter assumption allows to interpret $B(a)$ as the division rate in the time-scale defined by the size (and not the physical time), meaning that the cell uses its growth as a clock. Also, the increments of dividing cells are then mutually independent and distributed according to the density 
\[f_B(a)=B(a) \exp\Big(-\int_0^a B(s)\ ds\Big).\]
In the present paper, we will often consider functions of the form $g(x)= r\, x^\gamma$ and $B(a)=R\,  a^\eta $, for $r,R,\gamma,\eta>0$. \\

It is then possible to define an individual-centered stochastic process as in Section \ref{sec:modele-mitose} (see \cite[(2.10)-(2.11) page 25]{Doumic2023individual}) and to show that it converges in large population, with similar arguments as before, to a deterministic limit given by \eqref{EDP:adder}. 

Note that when $\kappa$ admits a density with respect to the Lebesgue measure and $\kappa(d\theta)=h(\theta)d\theta$, \tvc{the system} \eqref{EDP:adder} can be rewritten as:
\begin{align}\label{eq:adder-model-h}
    &    \partial_t n(t, a, x)  + \partial_a \big(g(x) n(t,a, x) \big) + \partial_x \big(g(x) n(t, a, x) \big) + B(a)g(x) n(t, a, x) = 0,\ \\
    & \hspace{2in} t\ge 0, x >0, a >0, \nonumber\\
    & g(x) n(t, 0, x) = 2 \int_0^\infty g(y)\int_0^\infty B(a)n\left(t, a, y\right) \frac{1}{y} h\left( \frac xy\right)da dy,\ \quad t\ge 0, x >0,\label{eq:adder-model-h2}\\
      &  n(0,a,x) = n_0(a,x),\qquad g(0)n(t,a,0) = 0. \nonumber
\end{align}
Under $H_0$, when $\kappa=\delta_{1/2}$, the boundary condition becomes (see \cite{DoumicOlivierRobert2020}):
\begin{align}
\label{eq:adder-model-h-dirac}
 g(x) n(t, 0, x) = 4 \int_0^\infty g(2x) B(a)n\left(t, a, 2x\right) da,\ \quad t\ge 0, x >0.
\end{align}

The eigenvalue problem associated to \tvc{the system} \eqref{EDP:adder} is now stated \tani{as}
\begin{align}
 &    \lambda N(a, x) + \partial_x \big(g(x) N(a, x) \big)  + \partial_a \big(g(x) N(a, x) \big) + g(x) B(a) N(a, x) =0,  \quad x > a >0, \label{eq:model-adder-stationary}\\
 &    g(x)N(0, x) = \tvc{2} \int_0^\infty \int_0^1 g\big(\frac{x}{\theta}\big)B(a) N\big(a, \frac{x}{\theta}\big) \frac{\kappa(d\theta)}{\theta} \, da, \quad x>0, \label{eq:model-adder-stationary2}\\
 &    \int_0^\infty \int_0^x N(a, x)da dx = 1. \nonumber
\end{align}
For smooth and bounded functions $g$ and $B$, the results of \cite{Doumic07} can be adapted, and they are stated in the case $\kappa=\delta_{1/2}$ by \cite{DoumicOlivierRobert2020}. 
\cite{Gabriel2019steady} studied the eigenvalue problem \eqref{eq:model-adder-stationary} when $g(x)=x$ with fairly general division rate $B(a)$ and fragmentation kernel $\kappa(d\theta)$, and obtained existence and uniqueness of the solution. \vince{The proof of existence and uniqueness of the eigenvalue problem can be extended to the case $g(x)=r \, x^\gamma$ and $B(a)=R \, a^\eta$, see Appendix \ref{app:existenceunicite-vp}, a model close to \cite{Michel06} for a size-structured population.}
%provided $\eta+1-\gamma>0$, meaning that the divisions compensate the increased growth speed in the case where $ \gamma>0$. 
When $\kappa$ has a density with respect to the Lebesgue measure, the convergence to $N(a,x)$ of the correctly rescaled solution of \eqref{EDP:adder} in long time can be established using entropy methods (e.g. \cite{Doumic07,DoumicOlivierRobert2020,Gabriel2019steady}).\\

We denote by $N_\delta$ the stationary density in the case of equal mitosis ($H_0$), when $\kappa = \delta_{1/2}$. In this case, the boundary condition becomes:
\begin{align}
g(x)N_\delta(0,x) = 4\int_0^{\infty} g(2x)B(a)N_\delta(a, 2x)da, \ x>0. \label{eq:N-delta_1/2-initial} 
\end{align}\tvc{Notice that when $\kappa=\delta_{1/2}$ and $\gamma=1$, oscillatory regimes may appear, see \cite{bernarddoumicgabriel2019,Gabriel2019steady}.}

\section{Reformulations of goodness-of-fit tests}\label{sec:goodness-of-fit}

\tani{In this section, we shall define $\mathcal{P}$ as the set of probability measures and $\mathcal{N}$ as the set of finite positive measures admitting a density on $\mathbb{R}^+$.}

\subsection{Reformulation of the goodness-of-fit test for the size-structured mitosis model}  \label{sec:test-procedure1}

\tani{We observe $n$ independent and identically distributed (i.i.d) random variables $Z_1, \dots, Z_n$ from the stationary density $N$. We want to test 
\[
H_0: \kappa(d\theta) = \delta_{1/2}(d\theta) \quad\text{ vs. }\quad H_1:  \kappa(d\theta) \neq  \delta_{1/2}(d\theta).
\]
%with $h(\theta)$ a density of probability supported on $[0,1]$ symmetric and centered at $1/2$.
}
Note that under the null hypothesis $H_0$, Equation \eqref{eq:model} corresponds to the equation for equal mitosis with constant division rate $R$ \tvc{\cite[Section 4.1.1 p.83]{Perthame07}}:
\begin{equation}\label{eq:equalmitosis}
 D(x)+ 2R\,  N(x)=4R\, N(2x), \quad x\geq 0.
\end{equation}

The idea of our statistical procedure is instead of performing our test on the initial distribution $\kappa$, we perform the test on the stationary density $N$ which is explicit under $H_0$ as shown in the next subsection.

\subsubsection{The stationary solution under $H_0$ }
Under $H_0: \kappa(d\theta) = \delta_{1/2}(d\theta)$,  the stationary solution $N$ has an explicit form $N_0$ specified in the following proposition.
%\tvc{IMPORTANT : To use the result by Perthame when $h=\delta_{1/2}$, we need $\alpha=1$. Can we assume this and erase all the $\alpha$.}

\begin{proposition}The unique positive solution of \eqref{eq:equalmitosis} is given by:
\begin{equation}\label{eq:N-explicit}
N_0(x)=\bar{N} \sum_{n=0}^{+\infty} (-1)^n \alpha_n e^{-2^{n+1}Rx}
\end{equation}with $\alpha_0=1$, $\alpha_n=\alpha_{n-1}\times 2/(2^n-1)$ and $\bar{N}$ a renormalizing constant.
\end{proposition}
See \cite[Lemma 4.1, Section 4.1.1]{Perthame07} for the proof of this proposition. In particular, it is proved there that $N_0$ is a positive function such that $N_0(0)=\lim_{x\rightarrow +\infty}N_0(x)=0$ and such that its derivatives also vanish at $0$ and infinity.

The computation gives that:
\begin{align*}
\alpha_n=\frac{2^n}{(2^n-1)\dots (2^1-1)},\qquad \sum_{n=0}^k (-1)^n \alpha_n=\frac{(-1)^k}{(2^k-1)\dots (2^1-1)}.
\end{align*}

\subsubsection{Injectivity}

\tani{ In order to perform our test on the stationary density $N$ instead on the initial probability distribution $\kappa$, we need the next proposition.

\begin{proposition}
In model (\ref{eq:model}), the following application $\mathcal{I}$ 
\begin{eqnarray*}
 \mathcal{I} :  \mathcal{P} &\rightarrow& {\mathcal{N}} \\
 \kappa &\mapsto& N
 \end{eqnarray*}
 is injective.
\end{proposition}

\begin{proof}
In the sequel, the notation $\mathcal{F}$ indicates the Fourier transform of any  integrable function $f$.
   % First, following \cite{HoangPhamNgocRivoirardTran2020},
   We denote
    $$
     \quad M(u):=e^uN(e^u), \quad \tilde D(u):= \partial_u(u \mapsto N(e^u))= e ^u N'(e^u)).
    $$
     Now, let us make the change of variable $x = e^u$ in Equation (\ref{eq:model}). We get that 
\begin{equation}\label{chang-variable}
\tilde D(u) + 2RM(u)= 2R(M\star \nu)(u),
\end{equation}
where $\nu$ is the image measure of $\kappa$ by the application $u \mapsto e^u$ and 
$\star$ denotes the convolution product
$$
(M\star \kappa)(u)= \int M(u-x) \nu(dx).
$$
Now taking the Fourier transform of both sides of \eqref{chang-variable}, we obtain
\begin{equation}\label{g_etoile}
\mathcal{F}(\nu)=\frac{\mathcal{F}(\tilde D)}{2R \mathcal{F}(M)}+1.
\end{equation}
% see Section 3.1.1 of \cite{HoangPhamNgocRivoirardTran2020} for full details. 
%Let us consider the first case, where $\kappa(d\theta)=  \delta_{1/2} ( d\theta)$. We have that $N=N_0$ with $N_0$ defined in \eqref{eq:N-explicit} if and only if $\kappa= \delta_{1/2}$.
%We now consider the second case, where $\kappa(d\theta)= h(\theta) d\theta$.
%We denote by  $^*$ the Fourier transform of any squared integrable function $f$ or of any probability measure.
   % First, following \cite{HoangPhamNgocRivoirardTran2020},
 %  We make the change of variable $\theta = e^u$ in Equation (\ref{eq:model})
    $$
%    g(u)= e^uh(e^u), \quad M(u):=e^uN(e^u), %\quad D(u):= \partial_u(u \mapsto N(e^u))= e ^u N'(e^u)),
    $$
%and get that 
%\begin{equation}\label{g_etoile}
%g^*=\frac{ D^*}{2R M^*}+1,
%\end{equation}
% see Section 3.1.1 of \cite{HoangPhamNgocRivoirardTran2020} for full details.  
 Now, suppose that we have two stationary densities $N_1$ and $N_2$ such that $N_1=N_2$.  Then $M_1=M_2$ and $\tilde D_1=\tilde D_2$ and thus $\mathcal{F}(M_1)=\mathcal{F}(M_2)$ and $\mathcal{F}(\tilde D_1)=\mathcal{F}(\tilde D_2)$. 
It follows from (\ref{g_etoile}) that
$$
\mathcal{F}(\nu_1)=\mathcal{F}(\nu_2).
$$
Due to injectivity of the Fourier transform on the set of probability measures, we have $\nu_1=\nu_2$ and hence $\kappa_1=\kappa_2$.
\end{proof}
}

\subsection{Reformulation of the goodness-of-fit test for the adder model}

We observe $n$ i.i.d random variables $Z_1, \dots, Z_n$ from the stationary density $N$, once again we want to test 
\tani{
\[
H_0: \kappa(d\theta) = \delta_{1/2}(d\theta) \quad\text{ vs. }\quad H_1: \kappa(d\theta) \neq  \delta_{1/2}(d\theta).
\]

\subsubsection{Injectivity \label{sec:test-procedure2}}

%with $h(\theta)$ a density of probability supported on $[0,1]$ symmetric and centered at $1/2$.

Similarly to Section \ref{sec:test-procedure1} we would like to rewrite the test $H_0:  \kappa=\delta_{1/2}$  as $H_0: N=N_\delta$ with $N_\delta$ associated to Equation \eqref{eq:N-delta_1/2-initial}. To this end, we need the model to be injective which is the statement of the next proposition.

\begin{proposition}
In model \eqref{eq:model-adder-stationary}-\eqref{eq:model-adder-stationary2}, assume that
\begin{equation}\label{hyp-moment}
\int_0^{+\infty}\int_0^{+\infty} g(x) N(a,x)B(a)\ da\  dx<+\infty.\end{equation}
Then the following application $\mathcal{I}$ 
\begin{eqnarray*}
 \mathcal{I} :  \mathcal{P} &\rightarrow& {\mathcal{N}} \\
 \kappa &\mapsto& N
 \end{eqnarray*}
 is injective.
\end{proposition}
}

The assumption \eqref{hyp-moment} is common in the literature. Based on \cite{Gabriel2019steady}, we provide a proof \tani{of \eqref{hyp-moment}} in Appendix \ref{app:existenceunicite-vp} for the case $g(x)=r x^\gamma$ and $B(a)=R a^\eta$.

\tani{
\begin{proof}
Suppose that we have two stationary densities $N_1$ and $N_2$ such that $N_1=N_2$. We will show that the corresponding kernels $\kappa_1 = \kappa_2$ where $\kappa_1$ and $\kappa_2$ are the probability measures corresponding to $N_1$ and $N_2$ respectively. 

From the initial condition $\eqref{eq:model-adder-stationary2}$, we have
\begin{eqnarray*}
 g(x)N_1(0,x) - g(x)N_2(0,x) &=& 2 \int_0^\infty \int_0^{1} g\left(\frac x \theta \right) \left[N_1\left(a, \frac x \theta \right) \kappa_1(d\theta) - N_2\left(a,  \frac x \theta\right) \kappa_2(d\theta) \right] B(a) da \\
 0&=& 2 \int_0^\infty \int_0^{1} g\left(\frac x \theta \right)N_1\left(a, \frac x \theta \right)   \left[\kappa_1(d\theta) - \kappa_2(d\theta) \right] B(a) da. 
\end{eqnarray*}
Changing the variable $x=e^u$ and $\theta=e^v$, and denoting $\nu$ the image measure of $\kappa$ by the application $v \mapsto e^v$, we get
\begin{eqnarray*}
\int_{-\infty}^0 g(e^{u-v}) \left[\int_0^\infty N_1(a, e^{u-v}) B(a)da \right] (\nu_1(dv)-\nu_2(dv))= 0.
\end{eqnarray*}
The l.h.s of the above equation is in fact equal to the following convolution product
$$
[\nu_1-\nu_2]\star [v \mapsto g(e^v) \int_0^\infty N_1(a, e^{v}) B(a)da ](u).
$$
Consequently, taking the Fourier transform of the convolution product, we get using \eqref{hyp-moment}, 
$$
\mathcal{F}(\nu_1-\nu_2)=0,
$$
which entails that $\nu_1=\nu_2$ and finally $\kappa_1=\kappa_2$.
\end{proof}
}

\section{Testing procedures}\label{procedures}

We are given $n$ i.i.d random variables $Z_1, \dots, Z_n$ with common unknown density $N$.  We have shown that instead of testing $h$, our goodness-of-fit tests in both models, the simple one (in one dimension) and the adder one (in two dimensions), can be devised on the stationary density $N$. 
 Accordingly we shall test at level $\alpha\in (0,1),$
\[
H_0: N = N_{0,\delta} \quad\text{ vs. }\quad H_1:  N \neq N_{0,\delta},
\]
with $N_{0,\delta}$ being either $N_0$ or $N_\delta$,  the stationary densities associated with the kernel division $h=\delta_{1/2}$ according to the considered model. \vince{Functions $g$ and $B$ are assumed to be known.} 
To implement our testing procedure, we use the goodness-of-fit test proposed in \cite{fromontlaurent2006}. This procedure manages to tackle the problem of nonparametric goodness-of-fit testing in a density model. 
Assume that $N_{0, \delta}$ and $N$ are square-integrable. The test statistics  is based on an estimation of the squared $L^2$-norm of the difference between $N$ and $N_{0,\delta}$: 
\[
\|N - N_{0,\delta}\|_2^2 = \|N\|_2^2 - 2\langle N, N_{0,\delta} \rangle + \|N_{0,\delta}\|_2^2.
\]
The term $\langle N, N_{0,\delta} \rangle$ is unbiasedly estimated by $n^{-1}\sum_{i=1}^n N_{0,\delta}(Z_i)$. So, it remains to construct an estimator for $\|N\|_2^2$. Let us consider $S_D$ a linear subspace of dimension $D$ of $L^2(\R^d), d=\{1,2 \}$ and let $\{\varphi_{\lambda}\}_{\lambda \in \Lambda_D}$ be an orthonormal basis of $S_D$. The quantity $\|N\|_2^2$ can be estimated by the $U$-statistics

\begin{equation}\label{eq:estimator-theta}
    \hat N_D = \frac{1}{n(n-1)} \sum_{i\neq j = 1}^n  \sum_{l \in \Lambda_D } \varphi_l(Z_i)\varphi_l(Z_j),
\end{equation}
%\vince{je change la notation}
which expectation is equal to $\mathbb{E}(\hat N_D)= \sum_{l \in \Lambda_D }  | \langle \varphi_l, N \rangle|^2$. Thus an estimator of the squared distance $\|N - N_{0,\delta}\|_2^2$ is given by

\begin{equation}\label{eq:T_hat}
    \hat T_D = \hat N_D - \frac{2}{n}\sum_{i=1}^n N_{0,\delta}(Z_i) + \|N_{0,\delta}\|_2^2.
\end{equation}

Now, we denote by $t_D(u)$ the $(1-u)$-quantile of the distribution of $\hat T_D$ under the null hypothesis $H_0: N = N_{0,\delta}$ and we consider 

%\begin{equation}\label{eq:level-test}
%    u_\alpha = \usup{u \in (0,1)}\left\{ \Prb_{N_{0,\delta}}\left( \usup{D\in\mathcal{D}} (\hat T_D - t_D(u)) > 0 \right) \le \alpha \right\}. 
%\end{equation}
%\vince{
%Ce n'est pas plutôt
$$
 u_\alpha = \sup \left\{u\in (0,1):\quad  \Prb_{N_{0,\delta}}\left( \usup{D\in\mathcal{D}} (\hat T_D - t_D(u)) > 0 \right) \le \alpha \right\}. 
 $$
% ?}
We introduce the test statistics $T_\alpha$ and the test $\phi$ defined by 
\begin{eqnarray}\label{eq:test-stat}
    T_\alpha &=& \usup{D\in\mathcal{D}} (\hat T_D - t_D(u_\alpha)) \label{eq:test-stat}\\
    \phi &=& 1_{\{T_\alpha >0\}}  \label{phi},
\end{eqnarray}
where $\mathcal{D}$ is a family of resolution levels. We reject $H_0$ if $\phi = 1$.  
\medskip

This method amounts to a multiple testing procedure. We first construct $u_\alpha$ and obtain a collection  $\{\hat T_D - t_D(u_\alpha), D \in \mathcal{D} \}$. We then decide to reject the null hypothesis if, for some of the tests of the collection, the null hypothesis is rejected. Quantity $u_\alpha$ and quantiles $\{ t_D(u_\alpha), D \in \mathcal{D} \}$ are computed by Monte Carlo replications. Specific numerical details are given in Section \ref{simulations} about the basis $\{ \varphi_\lambda \}$, the family $\mathcal{D}$, the quantity $u_\alpha$ and quantiles $\{ t_D(u_\alpha), D \in \mathcal{D} \}$.
\medskip

The test procedure $\phi$ defined in (\ref{phi}), which was proposed by \cite{fromontlaurent2006}, enjoys nice theoretical properties that we shall present now.  Let us consider the one-dimensional case. %Indeed, the test achieves optimal rates in a nonparametric minimax setting. 
%We shall give now some details about those theoretical results. 
\cite{fromontlaurent2006} studied the test $H_0: N=N_0$ against alternatives expressed in terms of Besov bodies. We shall recall the definition of a Besov body. One considers a pair of compactly supported orthonormal wavelets $(\varphi, \psi)$ such that for all $J \in \mathbb{N}$, $\{\varphi_{J,k}= 2^{J/2} \varphi(2^J\cdot -k), k \in \mathbb{Z} \} \cup \{\psi_{j,k}= 2^{j/2} \psi(2^j\cdot -k), k \in \mathbb{Z},\ j\geq J \}$ is an orthonormal basis of $L^2(\R).$ Then the Besov body  $ B_{s, 2, \infty}(R)$ is defined as  
$$B_{s, 2, \infty}(R)= \Big\{ f \in L^2(\mathbb{R}):\quad \forall j \in \mathbb{N}, \sum_{k \in \mathbb{Z}} \beta_{j,k}^2(f) \leq R^2 2^{-2js} \Big\},$$ 
with $\beta_{j,k}(f)= \langle f, \psi_{j,k} \rangle$. The authors assume that  $H_1$ is expressed in the following nonparametric form:
\begin{eqnarray}\label{H1-F-L}
H_1: N\in \mathcal{B}_{s}(R,M) = B_{s, 2, \infty}(R) \cap \Big\{g:\ \|g\|_{\infty} \leq M \Big\}.
\end{eqnarray}
Given $\beta \in ]0,1[$,  \cite{fromontlaurent2006} aim at evaluating the uniform separation rate $\rho_n$ of an $\alpha$-level test of $H_0 : N=N_0$ against $H_1 : N \in  \mathcal{B}_{s}(R,M)$. The uniform separation rate $\rho_n$  is defined as  the fastest rate of decay to zero as $n \rightarrow \infty$ which garantees a power at least equal to $1-\beta$ for all alternatives $N \in \mathcal{B}_{s}(R,M)$ at a $L^2$-distance $\rho_n$ from $N_0$.  %Then we have the following result for the test procedure $\phi $ described in (\ref{eq:test-stat}) to test $N=N_0$ against $N \in H_1$ specified in (\ref{H1-F-L}).  
In other words:
$$
\rho_n = \inf \{ \rho >0, \; \forall N \in \mathcal{B}_{s}(R,M), \;  \| N-N_0\|_2 \geq \rho_n \Rightarrow \mathbb{P}_N(\phi =1) \geq 1- \beta \}.
$$
They obtained the following theorem.

\begin{theorem}[\cite{fromontlaurent2006}]\label{theo-fromont-laurent}
Let $\phi$ be the test statistic defined in (\ref{eq:test-stat}). Assume that $n \geq 16$ and $\mathcal{D}= \{ 2^J, 0 \leq J \leq \log_2(n^2/(\log \log n)^3\}$. Fix some $\beta \in ]0,1[$. For all $s>0, M>0$ and $ (\log \log n)^{s +1/2} (\log n)^{2s +1/2} / \sqrt{n}  \leq R \leq n^{2s} /(\log \log n)^{3s +1/2}  $, there exists some positive constant $C=C(s, \alpha, \beta, M, \| N_0 \|_{\infty})$ such that if $N$ belongs to the set $\mathcal{B}_{s}(R,M)$ then 
$$
\rho_n \leq C R^{1/(4s+1)}\left ( \frac{\sqrt{\log \log n}}{n}\right)^{2s/(4s+1)}.
$$
\end{theorem}
We shall make some remarks about Theorem \ref{theo-fromont-laurent}. 
The proposed test achieves the adaptive  minimax rate  $\left ( \frac{\sqrt{\log \log n}}{n}\right)^{2s/(4s+1)}  $ proved by \cite{ingster2000}. \textit{Adaptive} means that the test procedure does not use prior smoothness assumptions on $N$ while \textit{minimax} means that it is not possible to test $H_0$ against $H_1$  at a faster rate than $\rho_n$. In other words, if the distance between $H_0$ and $H_1$ is less than $\rho_n$ then any test has a sum of errors probabilities of the first and second kinds close to 1 (trivial power). On the other hand, if the distance is of order $\rho_n$, then testing can be done with prescribed error probabilities. 

%\tvc{Bien préciser les approximations par Monte-Carlo} \ha{Hà: je vais donner plus détails pour cette section.}

\section{Numerical results on simulated data}\label{simulations}

We first perform simulations on data that we have generated from the mitosis model \eqref{EDP1} or from the adder model \eqref{EDP:adder}. Several densities are selected for the alternative the measure $\kappa(d\theta) = h(\theta)d\theta$. Our purpose is to have some insight about the quality of the goodness-of-fit test of Section~\ref{procedures} in terms of level and power.

\subsection{Test on simulated data from the mitosis model}\label{sec:test-mitosis-model}
%\tani{Il faut préciser comment est choisi $\mathcal{D}$}

In this section, we aim to test the hypothesis $H_0: N = N_0$ vs $H_1: N \neq N_0$ for model \eqref{eq:model}, under which, assuming the null hypothesis $H_0$ holds, the distribution $N_0$ has an explicit form given by \eqref{eq:N-explicit}. To perform this test, we generate a dataset $(X_1, \ldots, X_n)$ based on four distinct alternative densities $N_{1j}$ for $j = 1, \ldots, 4$. %Note that each $X_i$ represents the size of cell $i$, corresponding to the first coordinate of $Z_i$. 
Each alternative density $N_{1j}$ is is obtained by numerically solving the system \eqref{EDP:adder} w.r.t to a specific division kernel $h_j$, where \( \kappa(\theta) = h_j(\theta)\, d\theta \) and \( B(x) = R = 1 \). The division kernels $h_j$ for $j = 1, \ldots, 4$ are defined as follows:

\begin{enumerate}
	\item The $\mathrm{Beta}(2,2)$ density: $h_1(x) = C_1 x(1-x)\id_{[0,1]}(x)$ with $C_1$ the renormalizing constant. 
	\item The uniform density: $h_2(x) = \id_{[0,1]}(x)$. 
	\item The truncated normal distribution on $[0,1]$ with mean $1/2$ and variance $0.25^2$: 
	\[
	h_3(x) = \frac{\phi\left(\frac{x - \mu}{\sigma} \right)}{\sigma \left(\Phi\Big(\frac{1-\mu}{\sigma} \Big) -  \Phi\Big(\frac{-\mu}{\sigma} \Big)\right)},\quad x\in [0,1],
	\]
	where $\mu = 0.5$, $\sigma = 0.25$ and $\phi(\cdot)$ and $\Phi(\cdot)$ are respectively the density and the cdf  of the standard normal distribution.
	\item A mixture of Gaussian densities on $[0,1]$: 
	\[
	h_4(x) = C_2\left(\exp\left\{-\frac{((x-0.9)^2}{0.1}\right\} +  \exp\left\{-\frac{(1-x-0.9)^2}{0.1}\right\}\right)\id_{[0,1]}(x),
	\]
	with $C_2$ the renormalizing constant. 
\end{enumerate}

The test statistics are constructed according to equation \eqref{eq:test-stat}. To build the estimator $\hat N_D$, we employ the Laguerre basis, defined as follows:
\[
\varphi_j(x):=\sqrt{2}L_j(2x)e^{-x},\quad x\in\R_+,
\]
where $L_j$, for $j = 0, \ldots, D-1$, are the Laguerre polynomials given by
\[
L_j(x):= \sum_{k=0}^j \begin{pmatrix} j \\ k \end{pmatrix} (-1)^k\frac{x^k}{k!}.
\]
It can be verified that $\int_0^{+\infty} \varphi_j(x) \varphi_k(x) , dx = \delta_{jk}$, so that $(\varphi_j)_{ j \ge 0}$ forms an orthonormal basis in $L^2(\mathbb{R}_+)$ as claimed before. The resolution level $D$ will be assumed to belong to $\mathcal{D} = \{3, \ldots, 20\}$.

The significance level of the test is set at $\alpha = 5\%$. To estimate the threshold $u_\alpha$ and the quantiles $t_D(u_\alpha)$, we perform 200 Monte Carlo replications. Similarly, the power of the test under each alternative is estimated using 200 simulations. We also report the empirical levels of the test (which should be close to $5\%$).

\begin{table}[htbp]
	\begin{center}
		\begin{tabular}{|c|c|c|c|c|c|}
			\hline
			\multicolumn{2}{|c|}{Alternative}  & $\mathrm{Beta}(2,2)$ & Uniform & Truncated Gaussian & Mixture Gaussian \\
			\hline
			\multirow{3}{*}{Power}   & $n$ = 100 & 0.42  & 0.74     & 0.45      & 0.915 \\ 
			\cline{2-6}
			& $n$ = 200 & 0.62  & 0.965    & 0.68      & 0.99  \\ \cline{2-6}
			& $n$ = 500 & 0.995 & 1        & 0.99      & 1     \\
			\hline
			\multirow{3}{*}{Estimated level}  &$ n$ = 100 & \multicolumn{4}{c|}{0.05} \\\cline{2-6}
			& $n$ = 200 & \multicolumn{4}{c|}{0.075} \\\cline{2-6}
			& $n$ = 500 & \multicolumn{4}{c|}{0.06} \\
			\hline
		\end{tabular}
		\caption{Test powers and estimated significance levels for model \eqref{eq:model} across various sample sizes. \label{tab:result-simu-mitosis-model}}
	\end{center}
\end{table}

Overall, across the tested models with four different division kernels, the proposed test demonstrates strong performance. The test power increases substantially and approaches 1 as the sample size $n$ increases from 100 to 500 in all cases, indicating the consistency of the test. Moreover, the test exhibits particularly high sensitivity to complex alternatives; for instance, under the mixture Gaussian alternative, the test achieves a power of 0.915 even with $n = 100$, and reaches 1 when $n=500$. In contrast, when the alternative corresponds to smoother division kernels, such as the truncated Gaussian or the Beta(2,2) distribution, the power increases at a slower rate. This may be due to the fact that the alternative densities associated with smooth division kernels are more difficult to distinguish from \( N_0 \). In such cases, the $L^2$ distance between $N_{1j}$ and $N_0$ is relatively small, making it challenging for the test to detect a difference when the sample size is not sufficiently large. As a result, more data is required for the test to accumulate enough statistical evidence, which in turn leads to a slower convergence of the power towards 1.  Besides, across all cases the estimated significance levels vary  around $5\%$. This indicates that the test is well-calibrated under the null hypothesis $H_0$. % \vinc{Les niveaux sont un peu fluctuants, peut-etre dire un mot ?}

\subsection{Test on data simulated from the adder model}\label{sec:simu-adder-model}

This section investigates our testing procedures using the adder model, as introduced in Section \ref{sec:adder-model}. In order to perform the test, we first place ourselves under $H_0$ and thus need to find $N_{\delta}$. To do so, we first compute $n(t,a,x)$ using equations \eqref{eq:adder-model-h} and \eqref{eq:adder-model-h-dirac} for each $t \ge 0$, with given functions $g(x)$ and $B(a)$ described below. Then, we approximate $N_{\delta}(a, x) \approx n(t, a, x)$ once the solution becomes stable with respect to the variable $t$, \textit{i.e.},

\[
\dfrac{1}{\Delta t}\sum_a \sum_x \vert n(t_n, a, x) - n(t_{n - 1}, a, x) \vert < \epsilon,
\]
where the summation is taken over the mesh grid, and $\Delta t = t_n - t_{n - 1}$. The $L^2$-norm of this solution is then approximated by
\[
\lVert N_{\delta} \rVert_2^2 = \int_{\mathbb{R}^2 }\vert N_{ \delta}(a, x) \vert^2 \, da \, dx \approx \sum_a \sum_x \vert N_{\delta}(a, x) \vert^2\, \Delta a \, \Delta x.
\]
Similarly, as in Section \ref{sec:test-mitosis-model}, we find the alternative densities $N_j$, $j = 1,\ldots, 4$, corresponding to the division kernels $h_j$ introduced in Section \ref{sec:test-mitosis-model}. These densities are obtained by numerically solving equation \eqref{eq:adder-model-h} for $n(t, a, x)$ using the same growth rate $g(x)$ and division rate $B(a)$. Once the alternative densities are obtained, we generate a dataset $Z_1, Z_2, \ldots, Z_n$ where $Z_i = (A_i, X_i)$, drawn from either $N_{\delta}$ or $N_j$ using the rejection sampling method.

To build the test statistic, we use tensorized Daubechies wavelets.  Let  \( (\varphi, \psi) \) be a compactly supported pair of orthonormal Daubechies wavelets. For any \( x \in \mathbb{R}^+ \) and any \( j \in \mathbb{N} \), we set
\[
\varphi_{jk}(x) = 2^{j/2} \phi(2^j x - k), \quad \psi_{jk}(x) = 2^{j/2} \psi(2^j x - k).
\]
Then, let us define the tensor-product wavelet \( \boldsymbol{\varphi} \) as
\[
\boldsymbol{\varphi}(x, y) = \varphi(x) \psi(y).
\]
Let \( \boldsymbol{j}  = (j_1, j_2) \in \mathbb{N}^2 \) and \( \bk = (k_1, k_2)  \in  \mathbb{Z}^2 \), we define
\[
\boldsymbol{\varphi}_{\boldsymbol{j} \bk }(x, y) = 2^{j_1/2} 2^{j_2/2} \varphi(2^{j_1}x - k_1) \varphi(2^{j_2}y - k_2).
\]
%If we choose \( j_1 = j_2 = j \), then
%\[
%\boldsymbol{\varphi}_{\boldsymbol{j} \bk }(x, y) = \boldsymbol{\varphi}_{j \bk }(x, y)= 2^{j} \varphi(2^j x - k_1) \varphi(2^j y - k_2).
%\]
Note that the computation of $\hat{N}_{D}$ (defined in (\ref{eq:estimator-theta})) requires the sum of all integers, but here we restrict the sum to be $[-10, 10]^2 \cap \mathbb{Z}^2$, i.e
\[
\hat{N}_{D} \approx \frac{1}{n(n - 1)} \sum_{i \neq j = 1}^{n} \sum_{(k_1, k_2) \in \mathbb{Z}^2, \vert k_1 \vert, \vert k_2 \vert \le 10} \bphi_{J\bk}(Z_i) \bphi_{J\bk}(Z_j).
\]
The values of resolution level $J$ (equivalent here to $D = 2^J$) are chosen among the set $\{3,4,\ldots, 10\}$. 

For the choice of the growth rate $g(x)$ and the division rate $B(a)$, we consider seven distinct combinations of the growth rate $g(x)$ and the division rate $B(a)$, as presented in Table~\ref{tab:case_test_g_B}: 

\begin{table}[htbp]
	\centering
	\setlength{\tabcolsep}{12pt} % Increase horizontal padding between columns
	\begin{tabular}{|l|c|c|}
		\toprule
		\textbf{Case} &  \boldmath$g(x)$ & \boldmath$B(a)$ \\ 
		\midrule
		1 &  1 & 1 \\ 
		2 & $\sqrt{x}$ & 1 \\  
		3 & $\sqrt{x}$ & $a^2$ \\  
		4 & 1 & $a^2$ \\  
		5 & $\sqrt{x}$ & $a$ \\  
		6 & $\sqrt{x}$ & $\sqrt{a}$ \\  
		7 & $x$ & $a^2$ \\  
		\bottomrule
	\end{tabular}
	\caption{Growth rates $g(x)$ and division rates $B(a)$ used in the adder model test cases.}
	\label{tab:case_test_g_B}
\end{table}

At the significance level $\alpha = 5\%$, we estimate $u_\alpha$ and $t_D(u_\alpha)$ using 500 Monte Carlo replications. The power of the test and its empirical level are also computed using these 500 simulations. The simulation results corresponding to the seven test cases are presented in Tables \ref{tab:power_level_p1}. % and \ref{tab:power_level_p2}.

\clearpage 

\begin{table}[htbp]
	\begin{center}
		\begin{tabular}{|c|c|c|c|c|c|c|}
			\hline
			{Case}	& \multicolumn{2}{|c|}{{Alternative}}  & $\mathrm{Beta}(2,2)$ & {Uniform} & {Truncated Gaussian} & {Mixture Gaussian} \\
			\hline %\hline 
			% Case 01: g(x) = 1, B(a) = 1
			\multirow{8}{*}{1}   &     \multirow{3}{*}{Estimated powers}   & $n = 100$ & 0.608  & 0.938     & 0.61      & 0.99 \\ 
			\cline{3-7}
			&    & $n = 200$ & 0.834  & 0.996    & 0.802      & 1  \\ \cline{3-7}
			&   & $n = 500$ & 0.994 & 1        & 1      & 1     \\
			\cline{3-7} & & $n = 1000$ & 1 & 1 & 1 & 1\\
			\cline{2-7}
			& \multirow{3}{*}{Estimated levels}  & $n = 100$ & \multicolumn{4}{c|}{0.04} \\\cline{3-7}
			&  & $n = 200$ & \multicolumn{4}{c|}{0.038} \\\cline{3-7}
			&  & $n = 500$ & \multicolumn{4}{c|}{0.056} \\ \cline{3-7}
			&  & $n = 1000$ & \multicolumn{4}{c|}{0.048} \\
			\hline \hline 
			
			% Case 02: g(x) = sqrt(x), B(a) = 1 
			\multirow{6}{*}{2}   & \multirow{3}{*}{Estimated powers}   & $n = 100$ & 0.416  & 0.852     & 0.402      & 0.93 \\ 
			\cline{3-7}
			& & $n = 200$ & 0.708  & 0.99    & 0.688      & 0.998  \\ \cline{3-7}
			& & $n = 500$ & 0.944 & 1        & 0.954      & 1     \\
			\cline{2-7}
			& \multirow{3}{*}{Estimated level}  & $n = 100$ & \multicolumn{4}{c|}{0.04} \\\cline{3-7}
			&  & $n = 200$ & \multicolumn{4}{c|}{0.06} \\\cline{3-7}
			&   & $n = 500$ & \multicolumn{4}{c|}{0.05} \\
			\hline \hline 
			
			% Case 03: g(x) = sqrt(x), B(a) = a^2
			\multirow{6}{*}{3}   &  \multirow{3}{*}{Estimated powers}   & $n = 10$ & 0.796  & 0.94     & 0.732      & 0.966 \\ 
			\cline{3-7}
			&    & $n = 20$ & 0.954  & 0.998    & 0.924      & 1  \\ \cline{3-7}
			&    & $n = 50$ & 1 & 1        & 1      & 1     \\
			\cline{2-7}
			& \multirow{3}{*}{Estimated levels}  & $n = 10$ & \multicolumn{4}{c|}{0.072} \\\cline{3-7}
			& & $n = 20$ & \multicolumn{4}{c|}{0.06} \\\cline{3-7}
			&  & $n = 50$ & \multicolumn{4}{c|}{0.038} \\ \cline{3-7}
			\hline \hline 
			
			% Case 04: g(x) = 1, B(a) = a^2
			\multirow{6}{*}{4}   &       \multirow{3}{*}{Estimated powers}   & $n = 10$ & 0.748  & 0.916     & 0.794      & 0.976 \\ 
			\cline{3-7}
			&  & $n = 20$ & 0.97  & 0.994    & 0.976      & 1  \\ \cline{3-7}
			& & $n = 50$ & 1 & 1        & 1      & 1     \\
			\cline{2-7}
			& \multirow{3}{*}{Estimated levels}  & $n = 10$ & \multicolumn{4}{c|}{0.048} \\\cline{3-7}
			&    & $n = 20$ & \multicolumn{4}{c|}{0.07} \\\cline{3-7}
			&    & $n = 50$ & \multicolumn{4}{c|}{0.06} \\
			\hline \hline 
			
			% Case 05: g(x) = sqrt(x), B(a) = a
			\multirow{6}{*}{5}   &     \multirow{3}{*}{Estimated powers}   & $n = 10$ & 0.492  & 0.762     & 0.474      & 0.854 \\ 
			\cline{3-7}
			&  & $n = 20$ & 0.802  & 0.966    & 0.714      & 0.978  \\ \cline{3-7}
			&  & $n = 50$ & 0.974 & 1        & 0.952      & 1     \\
			\cline{2-7}
			&  \multirow{3}{*}{Estimated levels}  & $n = 10$ & \multicolumn{4}{c|}{0.048} \\\cline{3-7}
			&  & $n = 20$ & \multicolumn{4}{c|}{0.052} \\\cline{3-7}
			&   & $n = 50$ & \multicolumn{4}{c|}{0.042} \\
			\hline \hline 
			
			% Case 06: g(x) = sqrt(x), B(a) = sqrt(a)
			\multirow{6}{*}{6}   &    \multirow{3}{*}{Estimated powers}   & $n = 10$ & 0.304  & 0.552     & 0.274      & 0.674 \\ 
			\cline{3-7}
			& & $n = 20$ & 0.438  & 0.766    & 0.452      & 0.878  \\ \cline{3-7}
			&  & $n = 50$ & 0.786 & 0.988        & 0.808      & 0.996     \\
			\cline{3-7}   &  & $n = 100$ & 0.964  & 1    & 0.95      & 1  \\
			\cline{2-7}
			&  \multirow{3}{*}{Estimated levels}  & $n = 10$ & \multicolumn{4}{c|}{0.034} \\\cline{3-7}
			&  & $n = 20$ & \multicolumn{4}{c|}{0.05} \\\cline{3-7}
			&   & $n = 50$ & \multicolumn{4}{c|}{0.042} \\ \cline{3-7}
			&   & $n = 100$ & \multicolumn{4}{c|}{0.056} \\
			\hline \hline 
			
			% Case 07: g(x) = x, B(a) = a^2
			\multirow{6}{*}{7}   &  \multirow{3}{*}{Estimated powers}   & $n = 10$ & 0.76  & 0.934     & 0.696      & 0.968 \\ 
			\cline{3-7}
			&    & $n = 20$ & 0.932  & 0.992    & 0.936      & 0.998  \\ \cline{3-7}
			&    & $n = 50$ & 1 & 1        & 1      & 1     \\
			\cline{2-7}
			&  \multirow{3}{*}{Estimated levels}  & $n = 10$ & \multicolumn{4}{c|}{0.03} \\\cline{3-7}
			&   & $n = 20$ & \multicolumn{4}{c|}{0.042} \\\cline{3-7}
			&    & $n = 50$ & \multicolumn{4}{c|}{0.036} \\
			\hline
		\end{tabular}
	\end{center}
	\caption{Estimated powers and estimated levels w.r.t different sample sizes for case 1 to 7.}
	\label{tab:power_level_p1}
\end{table}

\clearpage

Figures \ref{fig:Na_Nx_p1} and \ref{fig:Na_Nx_p2} represent the distributions of $N_\delta(a,x)$ and $N_j(a,x)$ ($j = 1, \ldots, 4$), corresponding  to the division kernels $h= \delta_{1/2}$ and $h_j$ ($j = 1, \ldots, 4$), along the $a$ axis and $x$ axis, respectively.
\medskip 

Several remarks are in order.
\begin{enumerate}
\item Across all cases (from Case 1 to Case 7), the estimated significance levels vary in the range $[0.03, 0.072]$, centering around $5\%$. This indicates that the test is well-calibrated under the null hypothesis $H_0$. Although Case 3 reaches a slightly higher level (0.072), this is still acceptable given the small sample size ($n = 10$). Furthermore, we observe that the power consistently increases with the sample size across all cases, which reflects the consistency of the testing procedure.
\item For Cases 1 and 2, the test achieves high power (approaching 1) only for relatively large sample sizes $(n \ge 200)$. In contrast, for Cases 3 to 7, the power quickly converges to 1 even at smaller sample sizes ($n = 20$ or $50$). This difference can be explained by examining the shape of the alternative density functions $N_j(a, x)$, as illustrated in Figures \ref{fig:Na_Nx_p1} and~\ref{fig:Na_Nx_p2}. In Cases 1 and 2, the densities under the alternative hypotheses—corresponding to different division kernels (Beta(2,2), Uniform, Truncated Normal, and Gaussian Mixture)—tend to resemble the null density $N_\delta(a,x)$, making them harder to distinguish. As a result, the test requires relatively large sample sizes (e.g., $n = 500$ or $1000$) to achieve high power. In contrast, in Cases 3 to 7, the alternative densities are more distinct from the null density, and thus the test attains high power with much smaller samples, such as $n = 50$. This also explains why the sample sizes chosen for the experiments vary across different cases. For instance, in cases 1 and 2, a large sample size (e.g.,  $n = 500$ or $n = 1000$) is required to clearly observe the asymptotic behavior of the test power. In contrast, for other cases, the power reaches nearly 1 with just $n=50$, so there is no need to conduct experiments with larger sample sizes.
\item In all cases, the test performs best—i.e., yields the highest power—when the division kernel follows a {Gaussian Mixture} or {Uniform} distribution. Additionally, the specific forms of the functions $g(x)$ and $B(a)$ also influence test performance. Obviously, settings involving $g(x) = \sqrt{x}$ or $x$, combined with $B(a) = a^2$, tend to produce more complex alternative densities $N_j(a,x)$, which are more easily distinguishable from the null density. This is particularly observable in Cases 3, 4, 5, and 7.
\end{enumerate}

\smallskip 
\begin{figure}[htbp]
	\centering
	% Case 1
	\begin{subfigure}{0.45\textwidth}
		\includegraphics[width=\linewidth]{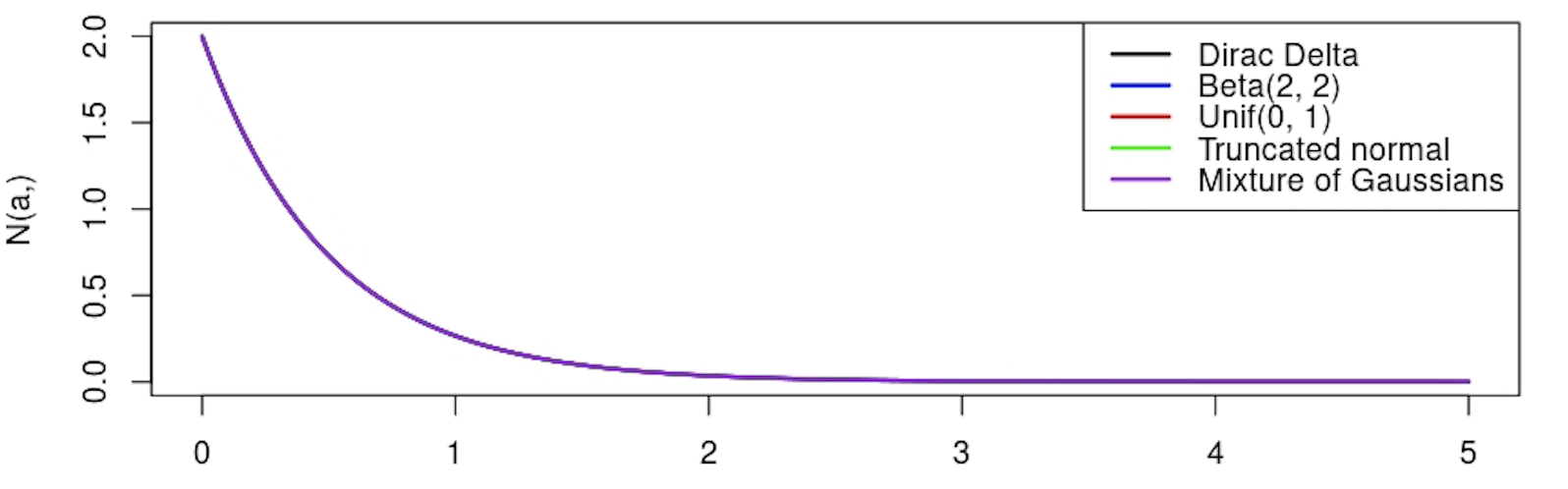}
		\caption{Case 1: $N(a,)$ w.r.t $g(x) = 1$, $B(a) = 1$.}
	\end{subfigure}
	\hfill
	\begin{subfigure}{0.45\textwidth}
		\includegraphics[width=\linewidth]{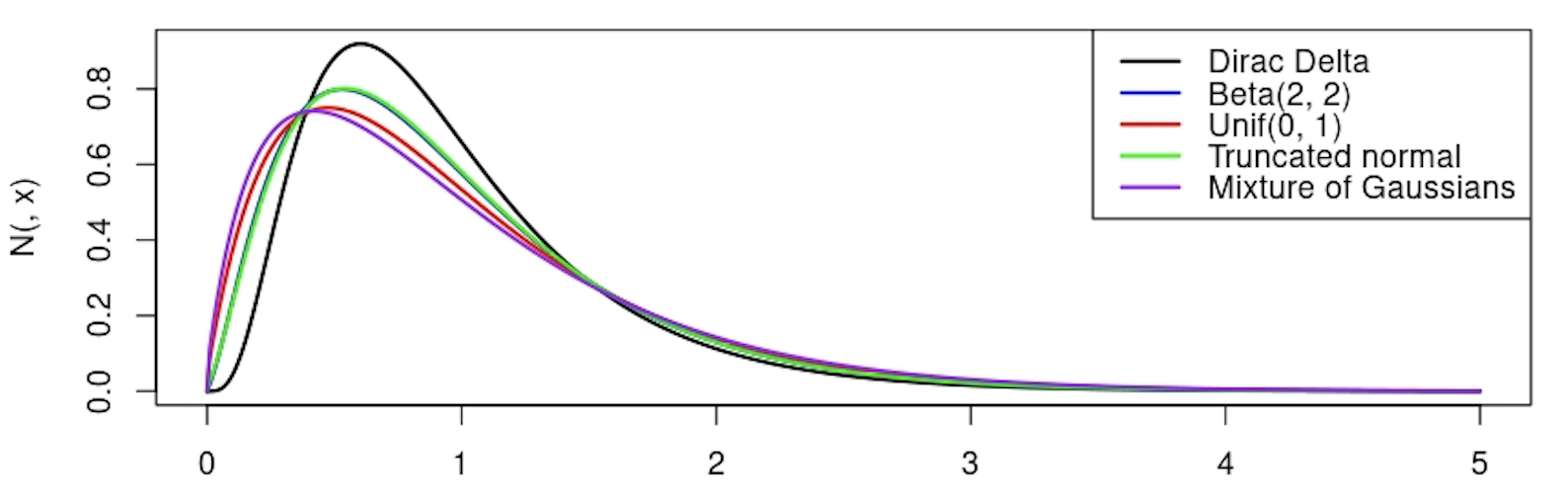}
		\caption{Case 1: $N(,x)$ w.r.t $g(x) = 1$, $B(a) = 1$.}
	\end{subfigure}
	\vspace{0.5cm}
	% Case 2
	\begin{subfigure}{0.45\textwidth}
		\includegraphics[width=\linewidth]{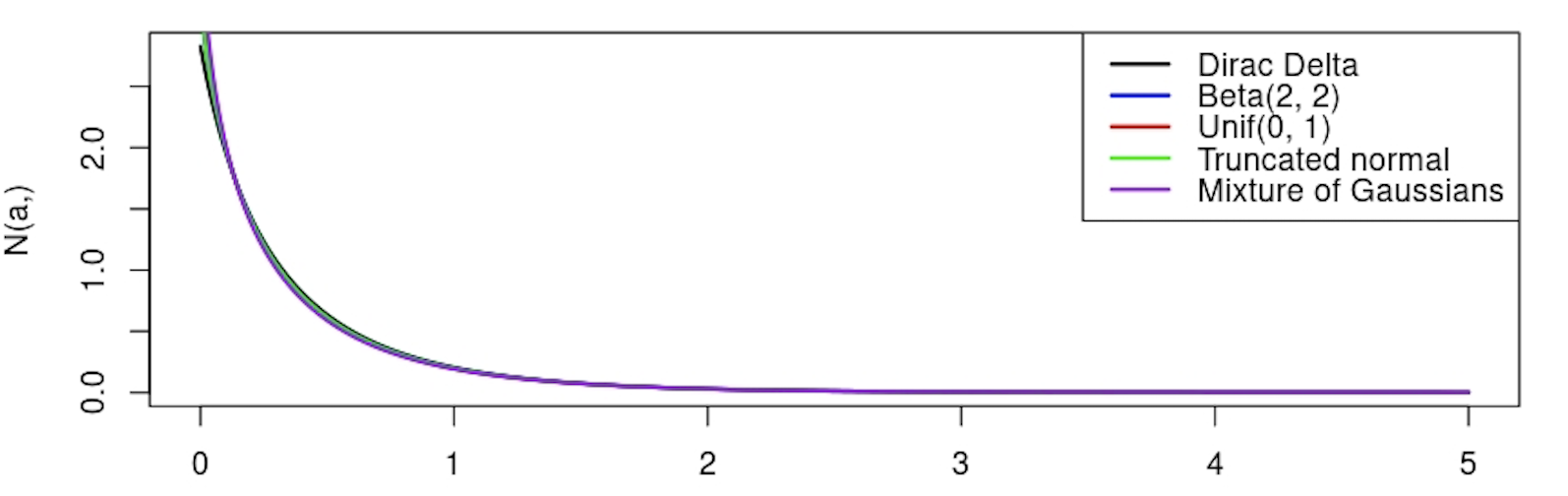}
		\caption{Case 2: $N(a,)$ w.r.t $g(x) = \sqrt{x}$, $B(a) = 1$.}
	\end{subfigure}
	\hfill
	\begin{subfigure}{0.45\textwidth}
		\includegraphics[width=\linewidth]{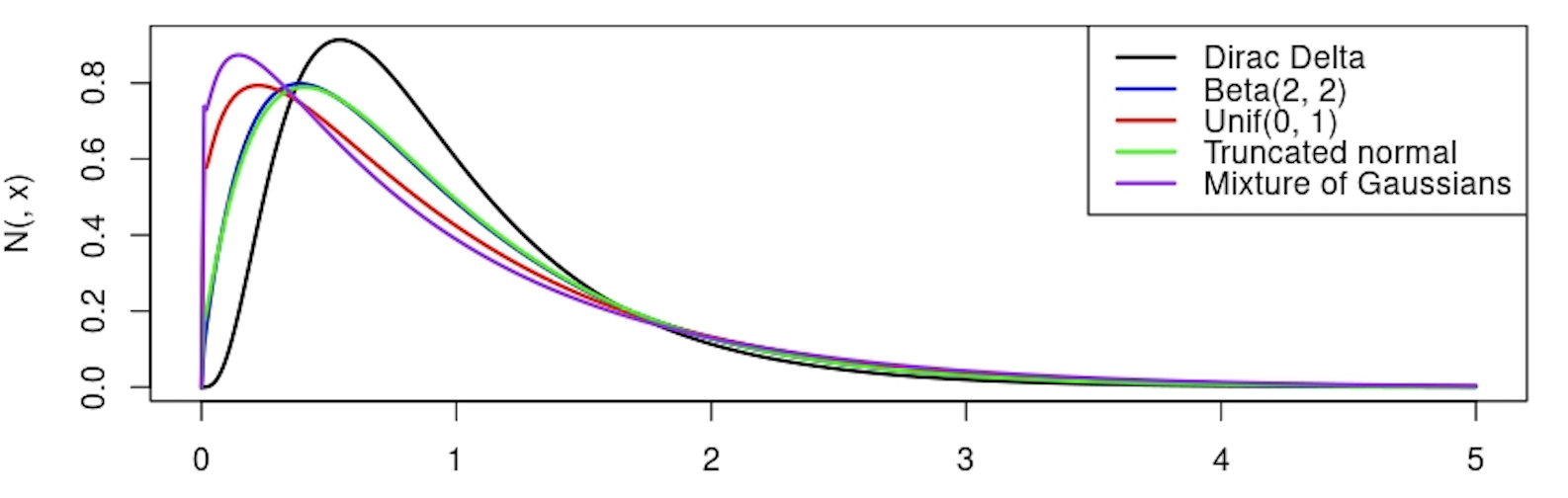}
		\caption{Case 2: $N(,x)$ w.r.t $g(x) = \sqrt{x}$, $B(a) = 1$.}
	\end{subfigure}
	\caption{Representations of $N_\delta$ and $N_j$ along each axis w.r.t the four division kernels—$\mathrm{Beta}(2,2)$, Uniform, Truncated Gaussian, and Gaussian Mixture—across Cases 1 to 2.}\label{fig:Na_Nx_p1}
\end{figure}

\begin{figure}[ht]
	\centering
%	% Case 1
%	\begin{subfigure}{0.45\textwidth}
%		\includegraphics[width=\linewidth]{Case1_Na.png}
%		\caption{Case 1: $N(a,)$ w.r.t $g(x) = 1$, $B(a) = 1$.}
%	\end{subfigure}
%	\hfill
%	\begin{subfigure}{0.45\textwidth}
%		\includegraphics[width=\linewidth]{Case1_Nx.png}
%		\caption{Case 1: $N(,x)$ w.r.t $g(x) = 1$, $B(a) = 1$.}
%	\end{subfigure}
%	\vspace{0.5cm}
%	% Case 2
%	\begin{subfigure}{0.45\textwidth}
%		\includegraphics[width=\linewidth]{Case4_Na.png}
%		\caption{Case 2: $N(a,)$ w.r.t $g(x) = \sqrt{x}$, $B(a) = 1$.}
%	\end{subfigure}
%	\hfill
%	\begin{subfigure}{0.45\textwidth}
%		\includegraphics[width=\linewidth]{Case4_Nx.png}
%		\caption{Case 2: $N(,x)$ w.r.t $g(x) = \sqrt{x}$, $B(a) = 1$.}
%	\end{subfigure}
%	\vspace{0.5cm}
	% Case 3
	\begin{subfigure}{0.45\textwidth}
		\includegraphics[width=\linewidth]{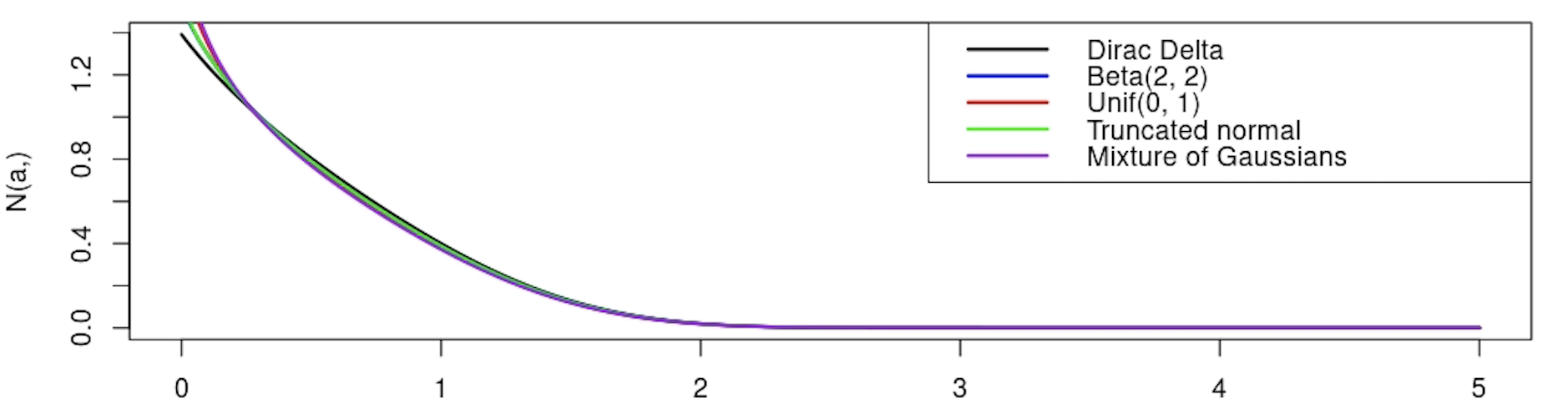}
		\caption{Case 3: $N(a,)$ w.r.t $g(x) = \sqrt{x}$, $B(a) = a^2$.}
	\end{subfigure}
	\hfill
	\begin{subfigure}{0.45\textwidth}
		\includegraphics[width=\linewidth]{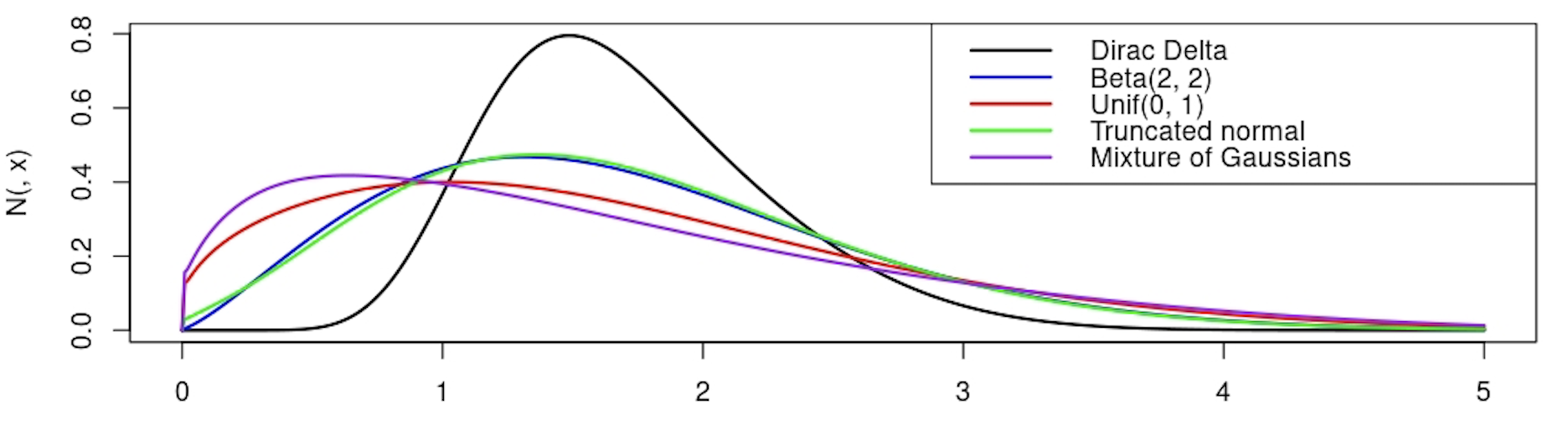}
		\caption{Case 3: $N(, x)$ w.r.t $g(x) = \sqrt{x}$, $B(a) = a^2$.}
	\end{subfigure}
	% Case 4
	\begin{subfigure}{0.45\textwidth}
		\includegraphics[width=\linewidth]{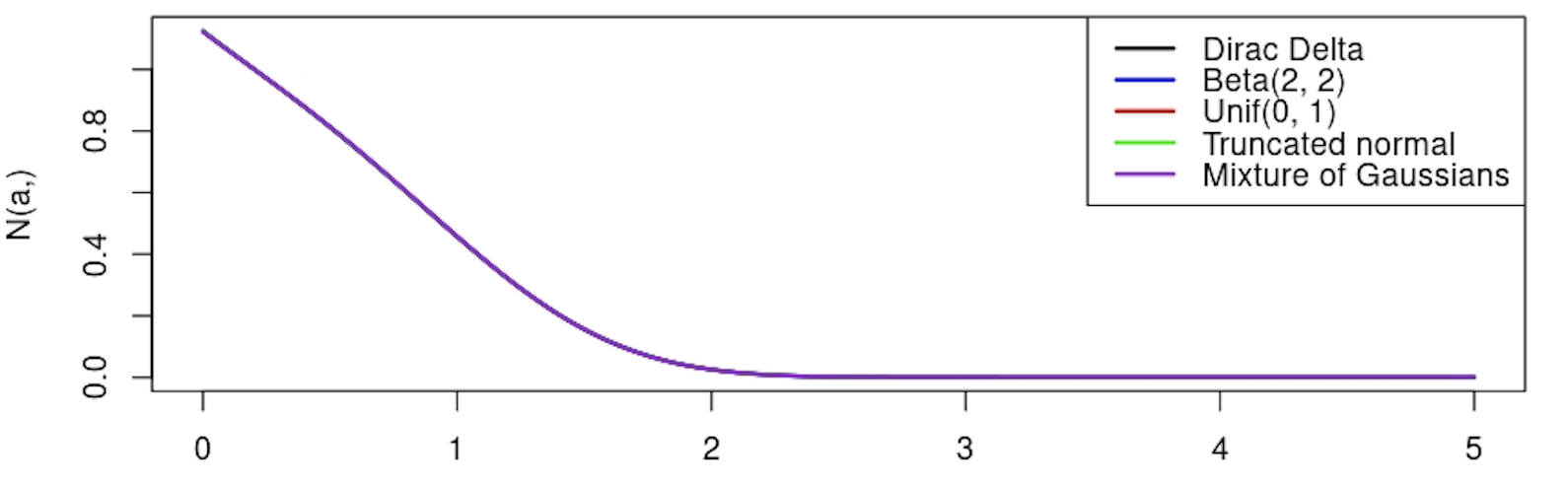}
		\caption{Case 4: $N(a,)$ w.r.t $g(x) = 1$, $B(a) = a^2$.}
	\end{subfigure}
	\hfill
	\begin{subfigure}{0.45\textwidth}
		\includegraphics[width=\linewidth]{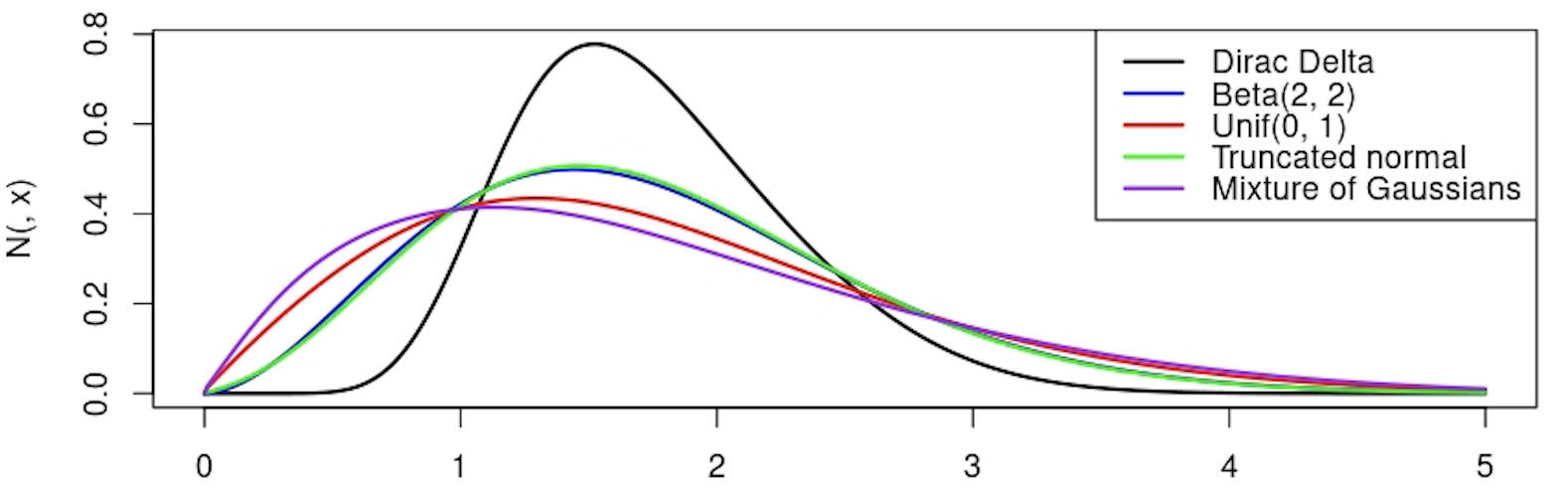}
		\caption{Case 4: $N(, x)$ w.r.t $g(x) = 1$, $B(a) = a^2$.}
	\end{subfigure}
	% Case 5
	\begin{subfigure}{0.45\textwidth}
		\includegraphics[width=\linewidth]{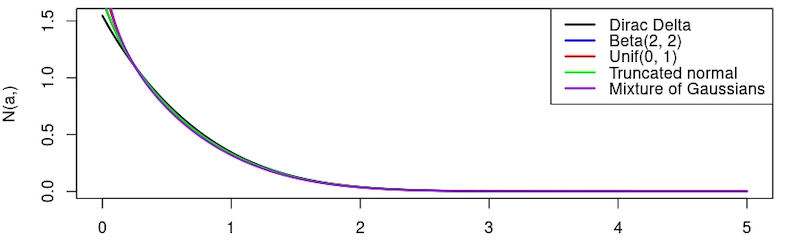}
		\caption{Case 5: $N(a,)$ w.r.t $g(x) = \sqrt{x}$, $B(a) = a$.}
	\end{subfigure}
	\hfill
	\begin{subfigure}{0.45\textwidth}
		\includegraphics[width=\linewidth]{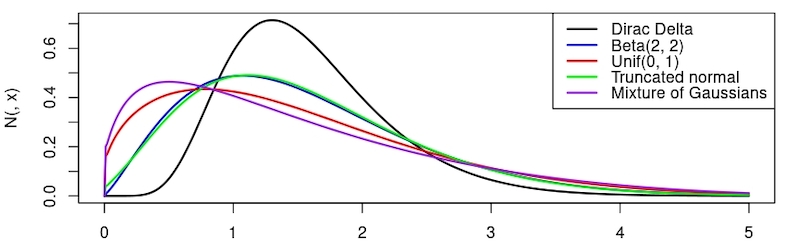}
		\caption{Case 5: $N(, x)$ w.r.t $g(x) = \sqrt{x}$, $B(a) = a$.}
	\end{subfigure}

	% Case 6
	\begin{subfigure}{0.45\textwidth}
		\includegraphics[width=\linewidth]{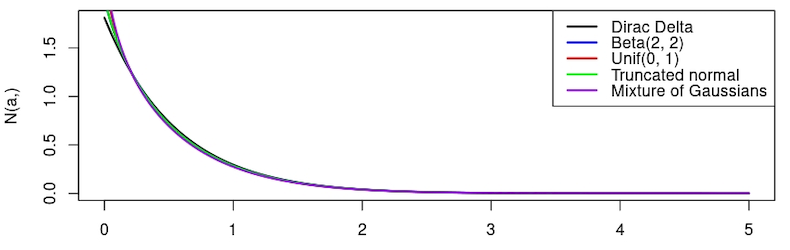}
		\caption{Case 6: $N(a,)$ w.r.t $g(x) = \sqrt{x}$, $B(a) = \sqrt{a}$.}
	\end{subfigure}
	\hfill
	\begin{subfigure}{0.45\textwidth}
		\includegraphics[width=\linewidth]{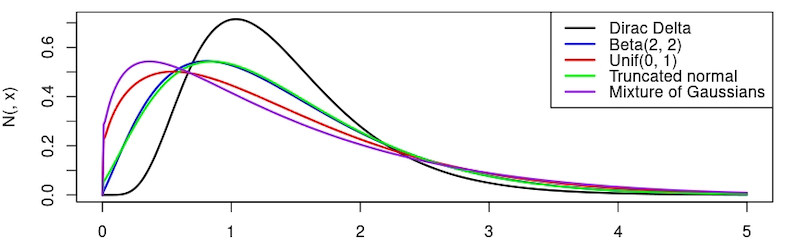}
		\caption{Case 6: $N(a,)$ w.r.t $g(x) = \sqrt{x}$, $B(a) = \sqrt{a}$.}
	\end{subfigure}
	% Case 7
	\begin{subfigure}{0.45\textwidth}
		\includegraphics[width=\linewidth]{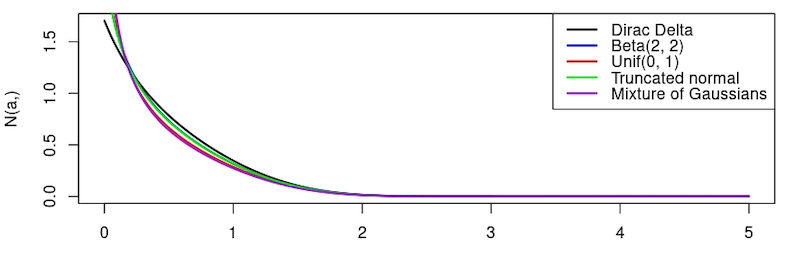}
		\caption{Case 7: $N(a,)$ w.r.t $g(x) = x$, $B(a) = a^2$.}
	\end{subfigure}
	\hfill
	\begin{subfigure}{0.45\textwidth}
		\includegraphics[width=\linewidth]{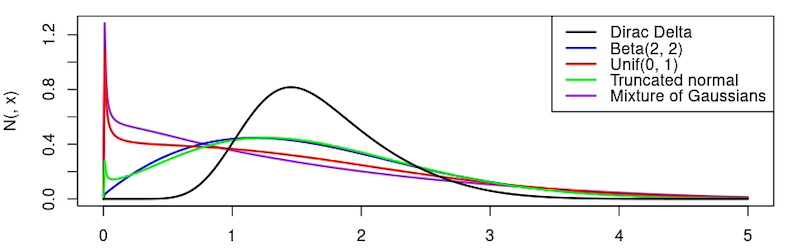}
		\caption{Case 7: $N(a,)$ w.r.t $g(x) = x$, $B(a) = a^2$.}
	\end{subfigure}
	\caption{Representations of $N_\delta$ and $N_j$ along each axis w.r.t the four division kernels—$\mathrm{Beta}(2,2)$, Uniform, Truncated Gaussian, and Gaussian Mixture—across Cases 3 to 7.}\label{fig:Na_Nx_p2}
\end{figure}

\newpage
\section{Testing with biological data}\label{sec:Ecoli}

\subsection{Data description}

\tvc{We are interested here is models for the division of \textit{Escherichia coli}. Of course, \textit{E. coli} have been much studied in recent years, see e.g. \cite{LydiaRobert14}, and experiments and model inference have shown in favor of the alternative hypothesis $H_1$. Our purpose here is to re-explore these questions with a rigorous statistical testing procedure. 
}

\tvc{We uses the density $N(x)$ that was reconstructed by \cite{Bourgeron14} and provided by these authors. This density is based on a real dataset of \textit{Escherichia coli} cells published in \cite{Stewart05} (see Figure \ref{fig:Nbio}). As the data by \cite{Stewart05} is only available in pooled form in the supplementary material of their paper, we generated $n = 10,000$ observations of sizes, $X_1, X_2, \ldots, X_n$, drawn from the density $N(x)$ provided in \cite{Bourgeron14}. 
}

\newpage 
\begin{figure}[htbp]
	\centering
	\includegraphics[scale = 0.75]{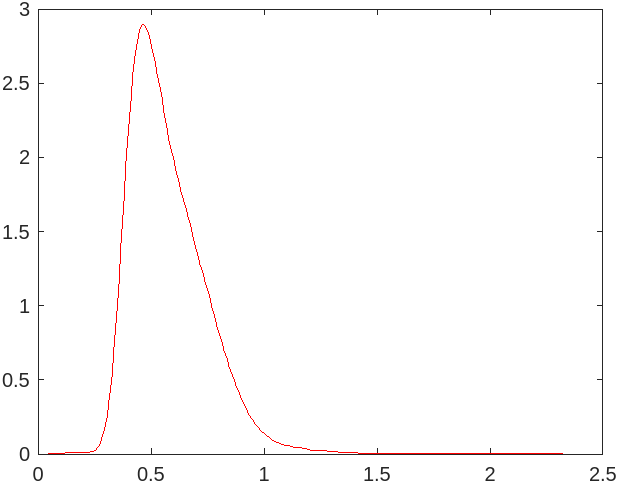}
	\caption{Density $N(x)$ of the real biological dataset.}
	\label{fig:Nbio}
\end{figure}

\tvc{From the data, we identify the best-fit model corresponding to the biological dataset. The models under consideration include:
\begin{itemize} 
	\item The mitosis model, which comprises two cases: one with a constant growth rate of $1$ and a division rate given by $B(x) = R$, and another with a general growth rate $g(x)$ and a division rate $B(x)$. 
	\item The adder model, in which we fit the marginal density $N(x)$ to the biological dataset. Indeed, the joint distribution of size and size-increment is not available, so we compare the empirical distribution with the distribution of sizes is obtained by marginalizing the distributions $N(a,x)$ from the model. 
\end{itemize}
Although \cite{LydiaRobert14} obtained on experimentations that the growth rate is linear and the division rate is not constant, we considered all the cases above to revisit the data and also check the sensitivity of our test to the model. 
}

\subsection{Fitting the mitosis model}

%We consider the following cases to find the best-fit model corresponding to the "real" dataset:

\subsubsection{Case of Constant Growth Rate $ g(x) = 1 $ and Division Rate $ B(a) = R $}  
\label{sec:mitosis-model-constant-case}

Under this configuration, the stationary density $ N(x) $ satisfies the simplified model defined in Equation~\eqref{eq:model}:  
%\[
%D(x) + 2R\,N(x) = 2R \int_0^\infty N(y)\, %h\left( \frac{x}{y} \right) \frac{dy}{y}, %\quad x \geq 0.
%\]  
\tani{
\begin{equation*}
 D(x) + 2R\  N(x) = 2R \int_0^1 N\left(\frac{x}{\theta}\right) \frac{ \kappa (d\theta)}{\theta}, \quad x \ge 0,
\end{equation*}
}
Under the null hypothesis $ H_0 $ (i.e., when $ N = N_0 $), this equation reduces to:  
\[
D(x) + 2R\,N(x) = 4R\,N(2x), \quad x \geq 0,
\]  
which admits an explicit stationary solution (Equation~\eqref{eq:N-explicit}):  
\[
N_0(x) = \bar{N} \sum_{n=0}^{\infty} (-1)^n \alpha_n e^{-2^{n+1} R x}.
\]  

Given a dataset $ X_1, \ldots, X_n$ sampled from $ N(x) $, we estimate the parameters $ R $ and $ \bar{N} $ in the expression of $ N_0(x) $ using the following approaches:  
\begin{enumerate}
	\item \textbf{Method of Moments}: Determine $ R $ and $ \bar{N} $ by solving the system:  
	\begin{equation}\label{eq:system-moment-N}
		\begin{cases}
			\lVert N_0 \rVert_{L^1} = \lVert N \rVert_{L^1}, \\
			\displaystyle\int_0^{\infty} x N_0(x)\,dx = \int_0^{\infty} x N(x)\,dx,
		\end{cases}
	\end{equation}
	where  
	\[
	\lVert N_0 \rVert_{L^1} = \frac{\bar{N}}{R} \sum_{n=0}^{\infty} (-1)^n \frac{\alpha_n}{2^{n+1}}, \quad \text{and} \quad \int_0^{\infty} x N_0(x)\,dx = \frac{\bar{N}}{R^2} \sum_{n=0}^{\infty} (-1)^n \alpha_n \frac{1}{2^{2n+2}}.
	\]
	The moment estimates are derived as:  
	\[
	\begin{cases}
		\hat{R}  = \dfrac{\sum_{n=0}^{\infty} (-1)^n \frac{\alpha_n}{2^{2n+2}}}{\sum_{n=0}^{\infty} (-1)^n \frac{\alpha_n}{2^{n+1}}} \dfrac{\lVert \hat{N} \rVert_1}{\int_0^\infty x \hat{N}(x)\,dx}, \\
		{\bar{N}= }\dfrac{\hat{R} \lVert \hat{N} \rVert_1}{\sum_{n=0}^{\infty} (-1)^n \frac{\alpha_n}{2^{n+1}}},
	\end{cases}
	\]
	where $ \hat{N} $ is the standard kernel density estimator constructed from $ X_1, \ldots, X_n$.
	
	\item \textbf{Least Squares Minimization}: Estimate $ R $ by minimizing the $ L^2 $-distance:  
	\[
	\hat{R} = \underset{R \in \mathcal{R}}{\arg\min} \lVert N_0 - N \rVert_2^2,
	\]
	where $ \mathcal{R} = \{0.01, 0.02, \ldots, 19.99, 20\} $. The renormalization constant $ \bar{N} $ is then computed as:  
	\[
	{\bar{N}= }  \left(\int_{0}^{\infty} \sum_{n=0}^{\infty} (-1)^n \alpha_n e^{-2^{n+1}\hat{R} x}\,dx\right)^{-1}.
	\]
\end{enumerate}

We first estimate $ R $ and $ \bar{N} $ using the biological dataset via both methods. A statistical test is then performed at a significance level $ \alpha = 0.05 $. The results are summarized in Table~\ref{tab:result_test_constant_case}.

\begin{table}[htbp]
	\centering
	\setlength{\tabcolsep}{12pt} 
	\renewcommand{\arraystretch}{1.5} 
	\begin{tabular}{|c|c|c|}
		\toprule
		\textbf{Parameter} &  \textbf{Moment Method} & \textbf{Least Squares} \\ 
		\midrule
		$ R $             & 1.732    & 1.4234   \\ 
		$ \bar{N} $       & 11.993   & 9.9032   \\
		$ u_\alpha $      & 0.05     & 0.0255   \\
		$ T_\alpha $      & 0.399    & 0.398    \\
		\bottomrule
	\end{tabular}
	\caption{Estimated parameters $\bar{N}$ and $R$, together with test statistics for assessing the fit of the mitosis model under constant growth and division rates.}
	\label{tab:result_test_constant_case}
\end{table}

Figure~\ref{fig:test-mitosis-model-constant-case} illustrates the densities obtained by parameters estimated from moment estimate and least square estimate with the density of real biological dataset. 

It can be observed that the parameter estimates for $R$ and $\bar{N}$ obtained via both the method of moments and the least squares minimization are relatively consistent in magnitude. Moreover, the hypothesis testing results derived from both approaches are in agreement, leading to the rejection of the null hypothesis $H_0$. Figure~\ref{fig:test-mitosis-model-constant-case} illustrates a clear discrepancy between the estimated densities (dashed and dot-dash lines) and the empirical density $N(x)$, indicating a poor fit. These observations suggest that the assumption of a constant division rate in the mitosis model fails to adequately capture the characteristics of the observed biological data.

%\begin{figure}[t]
%	\centering
%	\includegraphics[scale=0.45]{Case01_GridSearch_dataplot}
%	\caption{Densities constructed from moment estimates (blue dashed line) and least square estimate (blue dot-dash line) with the density $N$ in red.}
%	\label{fig:test-mitosis-model-constant-case}
%\end{figure}

\subsubsection{Case of general growth rate $g(x)$ and division rate $B(x)$}
\label{sec:mitosis-model-general-case}

Under this case and the null hypothesis, the stationary density $N = N_{0,\delta}$ satisfies the following equation

\begin{equation}\label{eq:Nbio_general_model}
	\partial_{x}(g(x)N_{0, \delta}(x)) + (\lambda + B(x))N_{0, \delta}(x) = 4B(2x)N_{0, \delta}(2x), \quad x \ge 0,
\end{equation}

Equation \eqref{eq:Nbio_general_model} is the corresponding eigenvalue problem of the PDE \eqref{EDP1} when $\kappa (d\theta) = \delta_{1/2}(d\theta)$, see for instance, \cite{DHRR}, \cite{Bourgeron14}. The parameter $\lambda$ here is the first eigenvalue of eigenvalue problem \eqref{eq:Nbio_general_model}. \\

Furthermore, we assume the following form of growth rate $g(x) = r x^\gamma$ and division rate $B(x) = x^\eta$.

%\newpage 
\begin{figure}[!ht]
	\centering
	\includegraphics[scale=0.45]{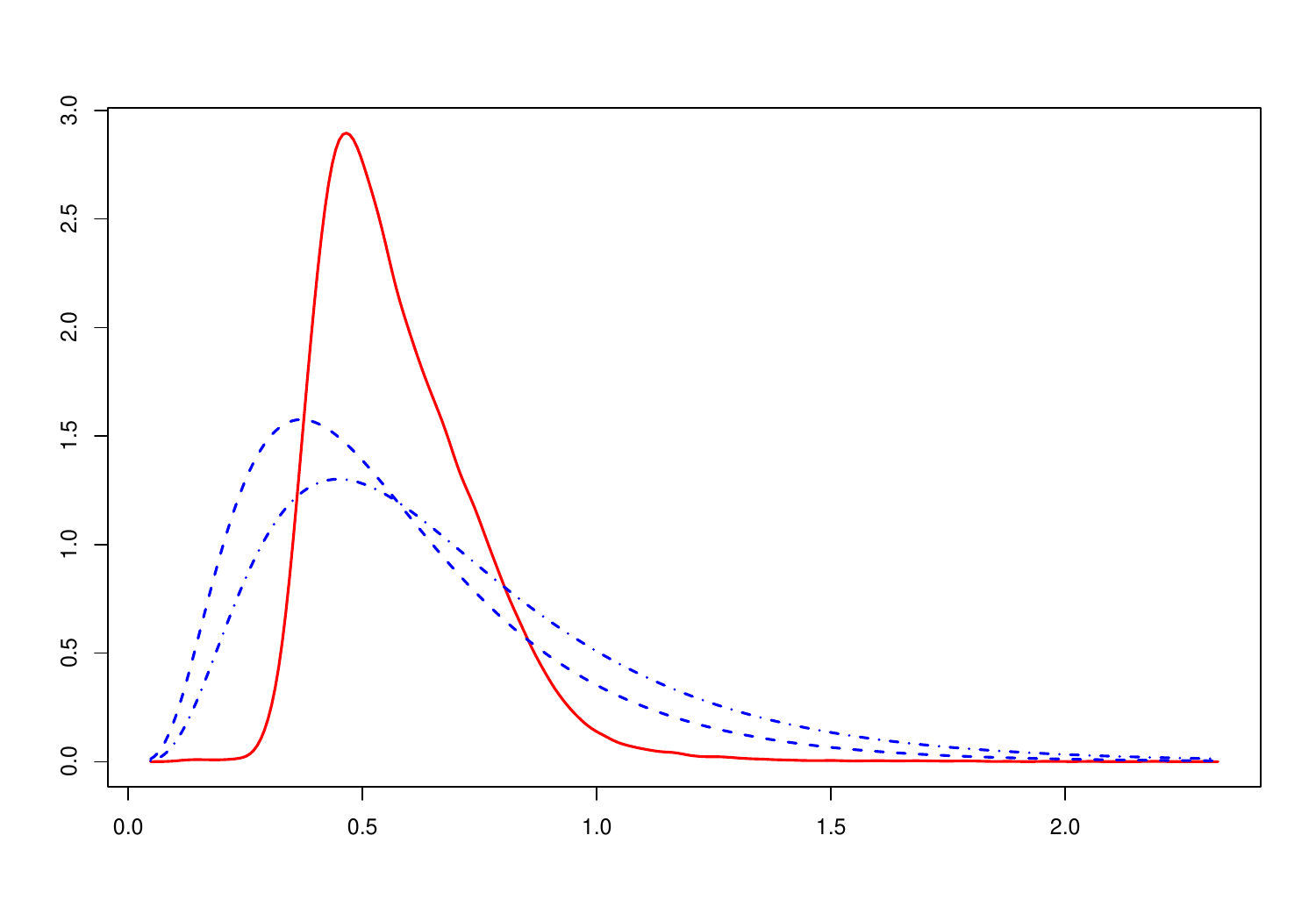}
	\caption{Densities constructed from moment estimates (blue dashed line) and least squares estimates (blue dot-dash line), together with the true density $N$ (red), for the mitosis model under constant growth and division rates.}
	\label{fig:test-mitosis-model-constant-case}
\end{figure}

To select appropriate form of $g(x)$ and $B(x)$ of the model fitted with the biological dataset, we conduct a grid search for each parameter: $ r, \gamma$, and $\eta$. We define the fixed grid as follows: \\

\begin{table}[h]
	\centering
	\begin{tabular}{|c|l|c|}
		\toprule
		\textbf{Parameter} & \hspace{1.3in}\textbf{Search Space} & \textbf{Number of Points} \\
		\midrule
		$r$ & $\{0.001, 0.01, 0.05, 0.01, 0.1, 0.5, 1, 5, 10, 20\}$ & $10$\\
		$\gamma$ & $\{0.5, 0.775, 1.05, \ldots, 5.45, 5.725, 6.0\}$ & $21$\\
		$\eta$ & $\{0.5, 0.825, 1.15, \ldots, 6.35, 6.675, 7.0, 7.5, 8.0, \ldots, 9.5, 10.0\}$ & $31$\\
		\bottomrule
	\end{tabular}
	\caption{Parameter search ranges for fitting the mitosis model with general growth rate $g(x) = r x^{\gamma}$ and division rate $B(x) = x^{\eta}$.}
	\label{tab:grid-search}
\end{table}

For each combination of $r, \gamma$, and $\eta$, we solve the PDE \eqref{EDP1} when $\kappa(d\theta) = \delta_{1/2}(\theta)$:
\[
\begin{cases}
	\partial_t n(t,x)+\partial_x \big( g(x) n(t,x)\big) +  B(x)n(t,x)=  4B(2x)n(t,2x), \\
	n(0,x) = n(x), n(t, 0) = 0. 
\end{cases}
\]
To get the stationary density $N_{0,\delta}$, then we compute the $L_2$-distance of $N_{0,\delta}$ with $N$, the density of real dataset. The parameters that give the best-fit to the biological dataset are those that minimize the $L^2$-norm
\[
(\hat r, \hat\gamma, \hat \eta) = \underset{r, \gamma, \eta}{\argmin} \, \|N - N_{0, \delta} \|_2^2. 
\]

We then obtain $\hat r = 0.05$, $\hat\gamma = 1.325, \hat \eta = 9.0$. Below is the plot that compares the density $N$ with the $N_{0, \delta}$ reconstructed from $\hat r, \hat \gamma$, and $\hat \eta$:

\begin{figure}[htbp]
	\centering
	\includegraphics[scale = 0.5]{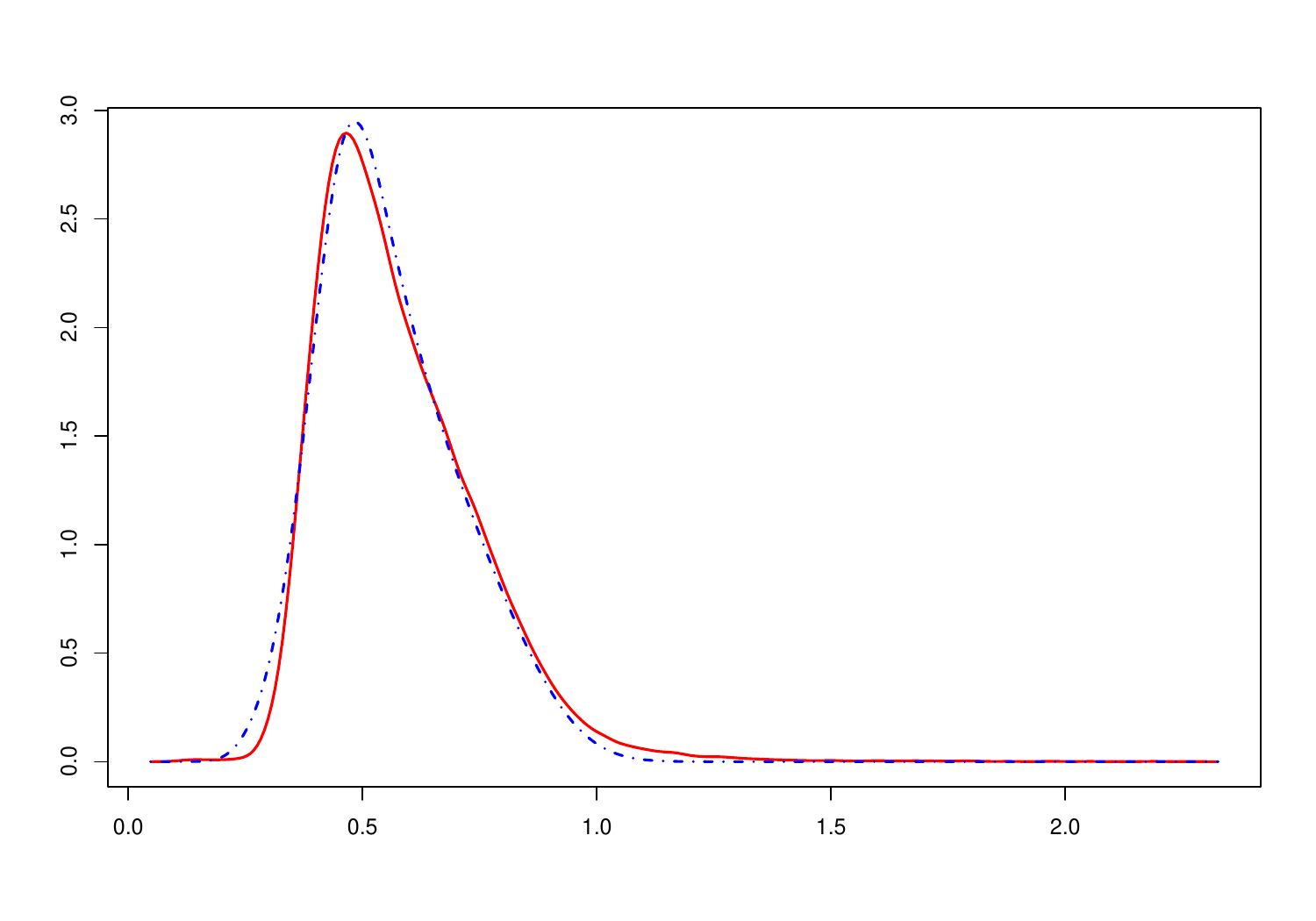}
    \caption{Comparison between the empirical density $N(x)$ (red line) derived from biological data and the fitted density $N_{0,\delta}$ (blue dotted line) obtained from the mitosis model with general growth rate $g(x) = r x^{\gamma}$ and division rate $B(x) = x^{\eta}$, using the parameter estimates $\hat{r} = 0.05$, $\hat{\gamma} = 1.325$, and $\hat{\eta} = 9.0$.}
	\label{fig:test-mitosis-model-general-case}
\end{figure}

Finally, we perform the statistical test with significant level $\alpha = 0.05$ to obtain $u_{\alpha} = 0.0355$ and $T_{\alpha} \approx -0.0381 \le 0$. Thus, we do not reject $H_0$ with significant level $\alpha = 0.05$.

\subsection{Fitting the adder model}

In this section, we investigate the fit of the adder model (see Section~\ref{sec:adder-model}) to the provided biological dataset. Since the available data only contains measurements of the density function $N(x)$, we focus on comparing and testing $ N(x) $ against the marginal density $ N_{\delta}(x) = \int N_{\delta}(a,x)\,da $.  

Following a similar approach to that used in Section~\ref{sec:mitosis-model-general-case}, we assume the growth rate and division rate follow the forms $g(x) = r x^\gamma $ and $B(a) = a^\eta$ where $r, \gamma, \eta$ are model parameters to be estimated. To determine the optimal values of these parameters, we conduct a grid search over a predefined set of candidate values. For each combination $ (r, \gamma, \eta) $, we numerically solve for the joint density $N_\delta(a,x)$ (see Section~\ref{sec:simu-adder-model}) and compute its marginal over size-increment $a$. We then evaluate the $L^2$-distance between the model-derived marginal density $N_\delta(x)$ and the empirical density $N(x)$. The parameter set that minimizes this distance is selected as the best fit. The search space for the parameters is shown in Table~\ref{tab:grid-search-2d}.

\begin{table}[htbp]
	\centering
	\begin{tabular}{|c|l|c|}
		\toprule
		\textbf{Parameter} & \hspace{1in}\textbf{Search Space} & \textbf{Number of Points} \\
		\midrule
		$ r $     & $ \{0.01, 0.1, 0.5, 1.0, 2.0, 3.0, 4.0, 5.0, 6.0, 7.0\} $ & 10 \\
		$ \gamma $ & $ \{5.0, 5.5, 6.0, 6.5, \ldots, 9.0, 9.5, 10.0\} $ & 11 \\
		$ \eta $  & $ \{0.1, 0.345, 0.59, \ldots, 5.0\} $ & 21 \\
		\bottomrule
	\end{tabular}
    \caption{Grid search parameter ranges for the Adder model with growth rate $g(x) = r x^{\gamma}$ and division rate $B(a) = a^{\eta}$.}

	\label{tab:grid-search-2d}
\end{table}

The optimal parameters values found through this procedure are $\hat r = 5.0$, $\hat \gamma = 9.0$, and $\hat \eta = 3.775$. Figure~\ref{fig:adder-N} illustrates the comparison between the empirical density $N(x)$ and the reconstructed marginal density $N_\delta(x)$ obtained using these estimated parameters. 

Finally, we perform the statistical test $H_0: N = N_\delta$ vs. $H_1: N \neq N_\delta$ with $\alpha = 0.05$. The test yields $u_\alpha - 0.0355$ and $T_\alpha = - 00261$, this implies that we do not reject the null hypothesis at the $\alpha  = 0.05$ level, supporting the validity of the adder model for this dataset.  

\newpage

\begin{figure}[t]
	\centering
	\includegraphics[scale=0.5]{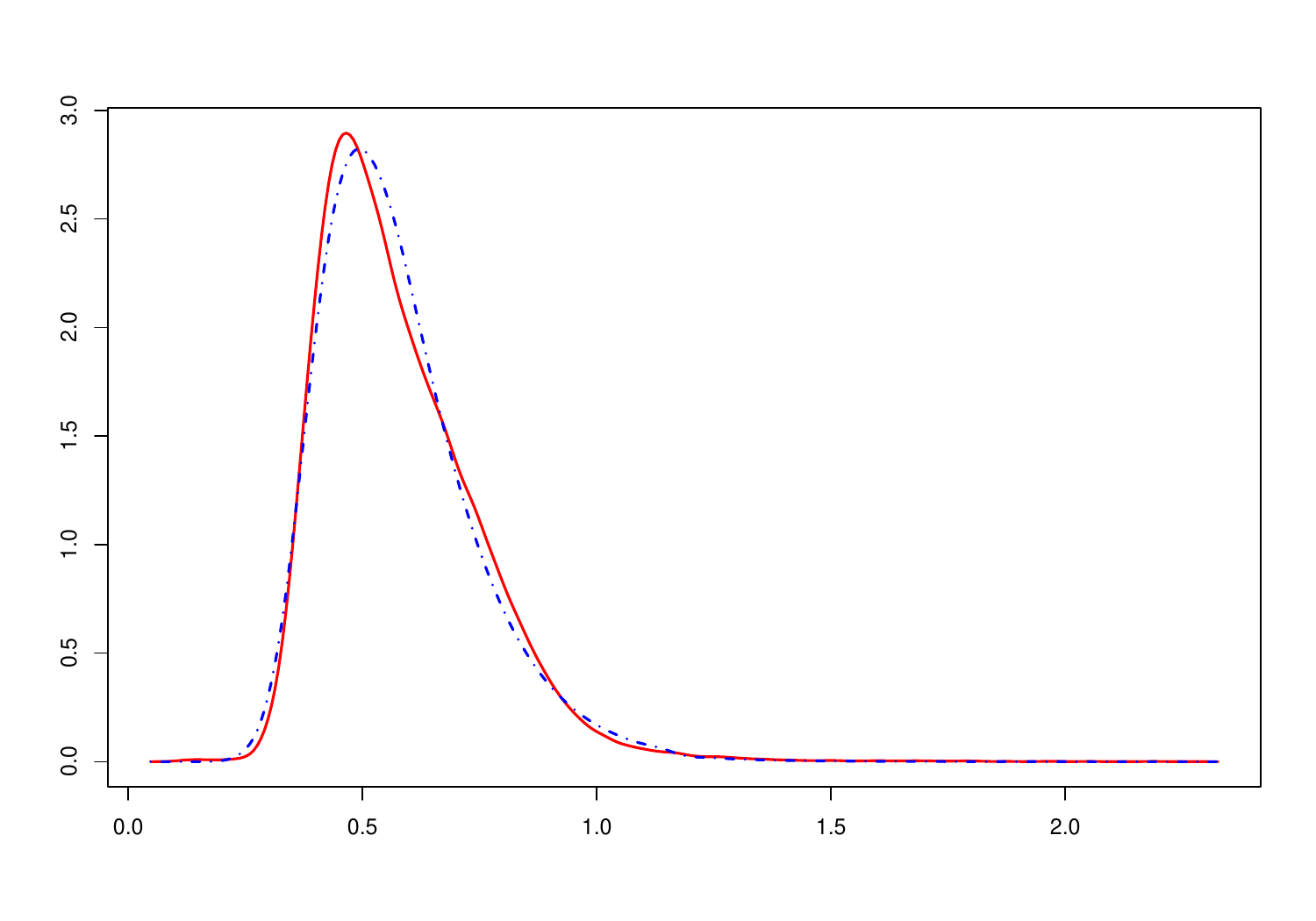}
	%\caption{Comparison of empirical density $ N(x) $  (in blue dotted lines) and fitted marginal density $ N_\delta(x) $ (in red) from the estimated parameters.}
    \caption{Comparison between the empirical density $N(x)$ (red line) and the fitted marginal density $N_{\delta}(x)$ (blue dotted line) obtained from the adder model with estimated parameters.}

	\label{fig:adder-N}
\end{figure}

\subsection{Conclusion}

\tvc{We have provided in this paper a statistical test adapted from 
the non-parametric goodness-of-fit approach of \cite{fromontlaurent2006} for the probability kernel appearing in the dynamics of dividing cells. In the framework considered here, observations consist of sizes of individuals drawn independently in $N(x,a)$, so the divisions are not directly observed. The inverse problem can be handled thanks to the injectivity between the division kernel $\kappa(d\theta)$ that is the object of the original test and the stationary age-size profile $N(x,a)$ of the PDE describing the population's evolution. Of course, if finer observation is available, direct computation of the distribution $\kappa(d\theta)$ can be made, following the work of \cite{Hoang2017}.} \\

\tvc{Data related with \textit{E. coli} are treated as a case study. We recover some results by \cite{LydiaRobert14}, namely that models with constant division rates fail to capture the characteristics of the data, while the mitosis model with growth rate $g(x)=0.05\, x^{1.325}$ and $B(x)=x^{9.0}$ one the one hand, and the adder model with $g(x)=5 \, x^{9.0}$ and $B(a)=a^{3.775}$ on the other hand both provide good fits of the first eigenfunction (see Figures \ref{fig:test-mitosis-model-general-case} and \ref{fig:adder-N}).
It is interesting to note that the null hypothesis (of equal division) is not significantly rejected in the chosen mitosis model, while it is for the adder model. Therefore an error in the model specification can lead to different conclusions, and it is not always easy to rule out one model for the other. 
Of course, \textit{E. coli} is considered here as an illustration of our testing procedure. We point out potential pitfalls, emphasizing the need of rigorous statistical procedures for treating the data. We have no intention to contradict previous results of biological studies of \textit{E. coli}, that are founded on richer and complementary datasets. \\
}

\bigskip 

\noindent \textbf{Acknowledgements: } Van Hà Hoang was invited at Universit\'e Gustave Eiffel in Autumn 2023 thanks to the `Campagne Actions Internationales 2023'. \tvc{We thank Marie Doumic and the anonymous Reviewer for interesting comments and questions that help us a lot to improve this work. We also thank Marie Doumic and Thibault Bourgeron for sharing datasets with us.} This work has been supported by Labex Bézout (ANR-10-LABX-58) and by the Chair ``Modélisation Mathématique et Biodiversité" of Veolia Environnement-Ecole Polytechnique-Museum National d’Histoire Naturelle-Fondation X. V.C.T. has been partially supported by the European Union (ERC, SINGER, 101054787). Views and opinions expressed are however those of the author(s) only and do not necessarily reflect those of the European Union or the European Research Council. Neither the European Union nor the granting authority can be held responsible for them.

\bibliographystyle{apalike}
%\bibliography{references}

\appendix

\section{Eigenvalue problem associated with the adder model}\label{app:existenceunicite-vp}

We consider the adder model with $g(x)=x^\gamma$ and $B(a)=a^\eta$, for $\gamma,\eta>0$, where we omitted the multiplicative constants for the sake of simplicity. Also, we consider the case where $\kappa(d\theta)=h(\theta)\, d\theta$. The case $\kappa(d\theta)=\delta_{1/2}$ has been checked in \cite{DoumicOlivierRobert2020}. Thus, let us consider the eigenvalue problem:

\begin{align}\label{eigenvalue_problem}
& \partial_x \big(x^\gamma N(a,x)\big)+ \partial_a \big( x^\gamma N(a,x)\big) + \big(\lambda+a^\eta x^\gamma\big) N(a,x)=0,\\
& x^\gamma N(0,x)= 2 \int_{0}^\infty \int_0^\infty a^\eta  y^{\gamma-1} h\Big(\frac{x}{y}\Big) N(a,y)\, dy\, da,\label{boundary_eigenvaluepb}\\
& - x^\gamma \partial_x \phi(a,x) -  x^\gamma \partial_a \phi(a,x) + \big(\lambda+a^\eta\big) \phi(a,x)=2 a^\eta x^{\gamma-1} \int_0^{+\infty} \phi(0,y) h\Big(\frac{y}{x}\Big)\, dy ,\\
& \int_0^{+\infty} \int_0^{+\infty} N(a,x)\, da\, dx=1,\qquad \int_0^{+\infty} \int_0^{+\infty} \phi(a,x) N(a,x)\, da\, dx=1.\label{eigenvalue_problem:app4}
\end{align}

\tvc{\begin{proposition}
    There is existence and uniqueness of a non-negative solution to the eigenvalue problem \eqref{eigenvalue_problem}-\eqref{eigenvalue_problem:app4}. For this solution,
    \begin{equation}
\int_0^{+\infty}\int_0^{+\infty} g(x) N(a,x)B(a)\ da\  dx<+\infty.\end{equation}
\end{proposition}
}

%Note that \eqref{boundary_eigenvaluepb} implies that for almost everywhere in $(a,x)$, $y\mapsto y^{\gamma-1}h(x/y) N(a,y)$ is integrable on $\R_+$ and therefore, this function converges to $0$ when $y\rightarrow +\infty$.

%\tvc{Assume that $\lim_{x\rightarrow +\infty} x^{\gamma+1} n(t,a,x)=0$.}
%Multiplying \eqref{EDP:adder} by $x$ and integrating in $(a,x)\in \R_+^2$, we obtain that:
%\begin{align}
%    \frac{d}{dt} \int_0^{+\infty} \int_0^{+\infty} x \, n(t,a,x)\ da\, dx = 
%\end{align}

\begin{proof}
Integrating \eqref{eigenvalue_problem} in $(a,x)\in \R_+^2$, and using \eqref{boundary_eigenvaluepb}, Fubini's theorem and the fact that $\int_0^1 h(\theta)d\theta=1$, we obtain:
\begin{equation}
    \lambda =  \int_0^{+\infty} \int_0^{+\infty} a^\eta y^\gamma N(a,y)\, dy\, da.
\end{equation}Note that in the case where $\gamma=1$, the total size $L(t)=\int_0^{+\infty} \int_0^{+\infty} x \, n(t,a,x)\ da\, dx $ in \eqref{EDP:adder} satisfies 
   $ \frac{d}{dt} L(t)=L(t)$, providing $\lambda=1$ (see \cite{Gabriel2019steady}), but this does not hold for $\gamma\not=1$.\\

 Following the steps of \cite{Gabriel2019steady} and defining $M(a,s)=N(a,a+s)$ where $s$ is the size at birth, we can see by computing the derivatives of $M$ and $N$ that \eqref{eigenvalue_problem} rewrites as:
 \begin{equation}
     \partial_a\big((a+s)^\gamma M(a,s)\big) + \big(\lambda + (a+s)^\gamma a^\eta\big) M(a,s)=0.\label{eigenvalue_problem_newvariable}
 \end{equation}For $s$ considered as a fixed parameter, \eqref{eigenvalue_problem_newvariable} is an ordinary differential equation whose solution writes generically as:
 \begin{equation}\label{generalsolutionODE}
     M(a,s)= M(0,s) \times \begin{cases} 
     \Big( \frac{s}{a+s}\Big)^2 e^{-\frac{a^{\eta+1}}{\eta+1} } & \mbox{ if }\gamma=1,\\
     \Big( \frac{s}{a+s}\Big)^\gamma e^{-\lambda\frac{(a+s)^{1-\gamma}-a^{1-\gamma}}{1-\gamma}} e^{-\frac{a^{\eta+1}}{\eta+1} } & \mbox{ if }\gamma\not= 1.
     \end{cases}
 \end{equation}
To obtain a solution of \eqref{eigenvalue_problem}-\eqref{boundary_eigenvaluepb}, it remains to find $M(0,s)$, which amounts in considering the boundary condition \eqref{boundary_eigenvaluepb} which becomes in the new variable $M(.,.)$:
 \begin{equation}
     s^\gamma M(0,s) = 2 \int_s^{+\infty} \int_0^y a^\eta y^\gamma \frac{1}{y}h\Big(\frac{s}{y}\Big) \, M(a,y-a) da\, dy.\label{boundary_newvariable}
 \end{equation}
 
Plugging \eqref{generalsolutionODE} in \eqref{boundary_newvariable} we obtain a fix point equation for $M(0,s)$. \\

\noindent \textbf{For $\gamma=1$:} Multiplying both sides by $s\geq 0$ and defining $f(s):=s^2 M(0,s)$, gives: 
\begin{align}
&     f(s) = 2 \int_s^{+\infty}\int_0^y a^\eta y^{-1} f(y-a)  e^{-\frac{a^{\eta+1}}{\eta+1}} \ da\ \frac{s}{y}h\Big(\frac{s}{y}\Big) dy\nonumber\\ 
\Leftrightarrow \qquad &     f(s) = 2 \int_0^1 \int_0^{s/\theta}  f(\alpha)  \Phi\Big(\frac{s}{\theta}-\alpha\Big) \ d\alpha\ h(\theta) d\theta,\label{eq:GabrielMartin}
\end{align}where \begin{equation}\label{def:Phi}\Phi(a)=a^\eta \exp\big(-\frac{a^{\eta+1}}{\eta+1}\big),\end{equation}
is a probability density on $\R_+$ corresponding to the lifetime distribution before division. This equation \eqref{eq:GabrielMartin} is similar to the one in \cite{Gabriel2019steady}. It can be rewritten as $f(s)=Tf(s)$ where $T$ is a transition operator linking the distribution of sizes at birth over successive generations. Notice that for $B(a)=a^\eta$, the survival function of the time before division is $\Psi(a)=\exp\big(-a^{\eta+1}/(1+\eta)\big)$ that satisfies the Assumptions of \cite[Theorem 1.1]{Gabriel2019steady}, in particular $\Psi(a)=O(a^{-k_0})$ for $a\rightarrow +\infty$ and for any $k_0>0$. The operator $T$ is a continuous linear operator on $L^1(\R_+)$ and Gabriel and Martin showed that there exists a unique solution in $L^1$ of this equation, that is the unique fixed point of $T$. \\
Moreover, 
\begin{align}
    \int_0^{+\infty} \int_0^{+\infty} x^\gamma a^{\eta}  N(a,x) \, da\ dx= &     \int_0^{+\infty} \int_0^{+\infty} x^\gamma a^{\eta}  M(a,x-a) \, \ind_{x\geq a}\, da\ dx\nonumber\\
    = &    \int_0^{+\infty} \int_0^{+\infty} (s+a)^\gamma a^{\eta}  M(a,s) \, ds\ da\nonumber\\
    = & \int_0^{+\infty} \int_0^{+\infty}  (s+a)^\gamma a^\eta \,  M(0,s)      \Big( \frac{s}{a+s}\Big)^2 e^{-\frac{a^{\eta+1}}{\eta+1} }  \, ds\, da\nonumber\\
    = & \int_0^{+\infty} s^2 M(0,s) \int_0^{+\infty}  (s+a)^{\gamma-2} a^\eta        e^{-\frac{a^{\eta+1}}{\eta+1} }  \, da\, ds,\label{etape2}
\end{align}
where we used the definition of $M$ and \eqref{generalsolutionODE}. In the right hand side, we recognize the Weibull distribution with shape parameter $\eta+1$, $a^\eta \exp(-a^{\eta+1}/(\eta+1))$, and therefore, for any $s>0$, the integral in $a$ is finite and can be bounded by $C s^{\gamma-2}$ in every case (whatever the sign of $\gamma-2$). As a result, \eqref{etape2} above can be upper bounded by $\int_0^{+\infty} s^\gamma M(0,s)ds$, which is finite by the properties of $T$ (see \cite[Lemma 2.2]{Gabriel2019steady}). This proves that \eqref{hyp-moment} holds in this case.\\

\noindent \textbf{For $\gamma\not=1$:} Proceeding similarly, with $g(s)=s^\gamma M(0,s)$,
\begin{align}
     & s^\gamma M(0,s) = 2 \int_s^{+\infty} \int_0^y a^\eta y^\gamma \frac{1}{y}h\Big(\frac{s}{y}\Big) \, M(0,y-a)  \Big(\frac{y-a}{y}\Big)^{\gamma} e^{-\lambda\frac{y^{1-\gamma}-a^{1-\gamma}}{1-\gamma}}  e^{-\frac{a^{\eta+1}}{\eta+1} }  da\, dy\nonumber\\
    \Leftrightarrow \qquad  & g(s) = 2 \int_s^{+\infty} \int_0^y  g(y-a) e^{-\lambda\frac{y^{1-\gamma}-a^{1-\gamma}}{1-\gamma}} \, a^{\eta} e^{-\frac{a^{\eta+1}}{\eta+1} }  da\  \frac{1}{y}h\Big(\frac{s}{y}\Big) dy\nonumber\\
    \Leftrightarrow \qquad  & g(s) = 2 \int_0^1 \int_0^{s/\theta}  g(\alpha) e^{-\lambda\frac{(s/\theta)^{1-\gamma}-(s/ \theta-\alpha)^{1-\gamma}}{1-\gamma}} \, \Phi\Big(\frac{s}{\theta}-\alpha\Big)  d\alpha\   \frac{1}{\theta} h(\theta) d\theta =: Sg(s), \label{def:S}
 \end{align}where $\Phi$ has been defined in \eqref{def:Phi}. Notice that for $\alpha \in (0,s/\theta)$ and whatever the sign of $1-\gamma$,
 \[0< \exp\Big(-\lambda\frac{(s/\theta)^{1-\gamma}-(s/ \theta-\alpha)^{1-\gamma}}{1-\gamma}\Big)  \leq 1.\]
 The existence and uniqueness of a nonnegative solution to \eqref{def:S} can be obtained from the existence of an eigenvector of $S$ associated to the eigenvalue 1. This can be obtained by an adaptation of the proof in \cite[Theorems 3.1, 3.3]{Gabriel2019steady}.
 \par Also a similar computation as in \eqref{etape2} leads to $\int_0^{+\infty} s^\gamma M(0,s)\, ds$ whose finiteness can be proven be extending Lemma 2.2 in \cite{Gabriel2019steady}.
\end{proof} 
\end{document}